%% file: cat-causal-full.tex
\documentclass{lmcs}
\pdfoutput=1

\usepackage{lastpage}
\lmcsdoi{15}{3}{15}
\lmcsheading{}{\pageref{LastPage}}{}{}%
{Apr.~06,~2018}{Aug.~09,~2019}{}

\usepackage{tikzfig}
\newcommand{\unbig}[1]{#1}
\input{preamble-lmcs.tex}
\usepackage{amsthm}
\usepackage{hyperref}

\newcommand{\CPM}{\ensuremath{\textbf{CPM}}\xspace}
\newcommand{\Rel}{\ensuremath{\textbf{Rel}}\xspace}
\newcommand{\MatRp}{\ensuremath{\textbf{Mat}(\mathbb R_+)}\xspace}

\newcommand{\SOCt}{\ensuremath{\textrm{SOC}_2}\xspace}
\newcommand{\SOCn}{\ensuremath{\textrm{SOC}_n}\xspace}

\newcommand{\past}{\ensuremath{\textbf{past}}\xspace}
\newcommand{\indep}{\mathop{\raisebox{-1.5pt}{\rotatebox{90}{$\models$}}}}
\newcommand{\evt}[1]{\ensuremath{\textsf{\footnotesize{#1}}}}

\usepackage[all]{xy}
\usepackage{braket}
\usepackage{graphicx}

\usepackage{braket}

\title{A categorical semantics for causal structure}

\author[A.~Kissinger]{Aleks Kissinger}  
\address{iCIS, Radboud University} 
\email{aleks@cs.ru.nl}  

\author[S.~Uijlen]{Sander Uijlen}  
\address{University of Oxford} 
\email{suijlen@cs.ru.nl}

\begin{document}
\girardnotation

\maketitle

\begin{abstract}
  We present a categorical construction for modelling causal structures within a general class of process theories that include the theory of classical probabilistic processes as well as quantum theory. Unlike prior constructions within categorical quantum mechanics, the objects of this theory encode fine-grained causal relationships between subsystems and give a new method for expressing and deriving consequences for a broad class of causal structures. We show that this framework enables one to define families of processes which are consistent with arbitrary acyclic causal orderings. In particular, one can define one-way signalling (a.k.a. semi-causal) processes, non-signalling processes, and quantum $n$-combs. Furthermore, our framework is general enough to accommodate recently-proposed generalisations of classical and quantum theory where processes only need to have a fixed causal ordering locally, but globally allow indefinite causal ordering.
  To illustrate this point, we show that certain processes of this kind, such as the quantum switch, the process matrices of Oreshkov, Costa, and Brukner, and a classical three-party example due to Baumeler, Feix, and Wolf are all instances of a certain family of processes we refer to as $\SOCn$ in the appropriate category of higher-order causal processes. After defining these families of causal structures within our framework, we give derivations of their operational behaviour using simple, diagrammatic axioms.
\end{abstract}


\tableofcontents

\section{Introduction}

Broadly, causal structures identify which events or processes taking place across space and time can, in principle, serve as causes or effects of one another. In the field of statistical causal inference, directed acyclic graphs (or generalisations thereof) have been used to capture that fact that certain random variables may have causal influences on others. That is, forcefully changing a certain variable (e.g., whether a patient takes a drug or placebo in randomised trial) may have an effect on another variable (e.g., whether the patient recovers). Intuitively, it may seem impossible to draw such causal conclusions without making such an intervention, however it has been shown in recent years that, under certain assumptions, one can draw causal conclusions from purely observational data~\cite{PearlBook}. For instance, the \textit{constraint-based approach} to causal discovery uses conditional independences present in statistical data to iteratively rule out possible causal relationships by removing edges from a (typically very large) directed graph. This sits at the heart of modern, graph-based causal discovery algorithms such as FCI~\cite{SpirtesGlymour}, and has been successful in a wide variety of real-world applications, including the study of ADHD~\cite{CausalADHD}, neural connectivity~\cite{CausalBrain}, gene regulation networks~\cite{Genes}, and climate change~\cite{Climate}.







In recent years, it has been asked whether, and to what extent, concepts from classical causal reasoning can be transferred to quantum theory, or more generalised theories of interacting processes. A major obstacle to employing classical techniques comes from the presence of \textit{quantum non-locality}, i.e., the possibility of correlations observed between distant, non-communicating agents which cannot be explained by a classical common cause. Indeed, Wood and Spekkens showed that quantum theory allows correlations which admit no faithful classical causal model~\cite{WoodSpekkens}. Roughly speaking, this means that any attempt to na\"ively reproduce quantum correlations with a classical causal model will necessarily include spurious causal relationships. Hence, there have been numerous attempts to extend classical causal models to quantum~\cite{PienaarBrukner,CostaShrapnel,AllenCommonCause} or even more general~\cite{HLP} models.


In the context of quantum theory, causal relationships between inputs and outputs to quantum processes have been expressed in a variety of ways. Perhaps the simplest are in the form of non-signalling constraints, which guarantee that distant agents are not capable of sending information faster that the speed of light, e.g., to affect each other's measurement outcomes \cite{BGNP}. Quantum strategies \cite{QGames} and more recently quantum combs \cite{QuantCircArch} offer a means of expressing more intricate causal relationships, in the form of chains of causally ordered inputs and outputs.  Furthermore, it has been shown recently that one can formulate a theory that is locally consistent with quantum theory yet assumes no fixed background causal structure \cite{ViennaIndef}. Interestingly, such a theory admits \textit{indefinite} causal structures. Namely, it allows one to express processes which inhabit a quantum superposition of causal orders.

Such processes are especially interesting for two reasons. First, indefinite causal structure seems to be an inevitable property of any theory that attempts to combine causal dynamics of general relativity with the irreducible non-determinism of quantum theory~\cite{HardyCausaloid}. Hence, indefinite causal structures can provide interesting `toy models' that exhibit quantum gravity-like features. Second, if physically realisable, processes with indefinite causal ordering can be exploited to gain an advantage in comminication problems, such as the `guess your neighbour's input' game~\cite{ViennaIndef} and computational tasks, such as single-shot quantum channel discrimination \cite{QSwitch}. Perhaps surprisingly, this phenomenon is not unique to quantum processes: it has been shown recently that, in the presence of three or more parties, causal bounds can be violated even within a theory that behaves locally like classical probability theory \cite{WolfClassicalMulti}. The key ingredient in studying (and varying) causal structures of processes is the development of a coherent theory of \textit{higher order causal processes}. That is, if we treat local agents (or events, laboratories, etc.) as first order causal processes, then the act of composing these processes in a particular causal order should be treated as a second-order process.

This paper aims to represent causal relationships within `black box' processes (which could be classical, quantum, or more general) in a uniform way, and provide some of the first tools for `generalised causal reasoning' which can be applied in both the classical and quantum contexts. To do this, we start from a category $\mathcal C$ of `raw materials'---whose morphisms should be thought of as processes without any causal constraints---and construct a new category $\Caus[\mathcal C]$ of first and higher-order processes which are consistent with certain given causal constraints.

For $\mathcal C$, we introduce a new kind of category called a \textit{precausal category}, which is a compact closed category with a bit of extra structure enabling us to reason about causal relationships between systems.
Most notably, precausal categories have a special \textit{discarding process} defined on each object, which enables one to say when a process is \textit{causal}, namely when it satisfies the following equation:
\[ \input{causal.tikz} \]
This has a clear operational intuition: 
\begin{center}
\medskip

\textit{If the output of a causal process is discarded, it doesn't matter which process occurred}.

\medskip
\end{center}
...or put a slightly different way:
\begin{center}
\medskip

\textit{The only influence a causal process has is on its output}.

\medskip
\end{center}
In particular, causality rules out the possibility of hidden side effects or the `backwards flow' of information from a given process.
While seemingly obvious, the causality condition, originally given by \cite{Chiri2} in the context of operational probabilistic theories, is surprisingly powerful. For instance, it is equivalent to the non-signalling property for joint processes arising from shared correlations \cite{Coecke2016}.

It is thus natural to consider the subcategory of $\mathcal C$ consisting of all causal processes. However, we take this a step further and consider not only (first-order) causal processes, but also higher-order mappings which preserve certain kinds of causal processes. This yields a $*$-autonomous category \CausC, into which the category of all first-order causal processes embeds fully and faithfully. We therefore call categories of the form $\Caus[\mathcal C]$ \textit{higher-order causal categories} (HOCCs). Our main examples start from the precausal categories of matrices of positive real numbers and completely positive maps, which will yield HOCCs of higher-order classical stochastic processes and higher-order quantum channels, respectively.

While categorical quantum mechanics \cite{AC1} has typically focussed on compact closed categories of quantum processes, we show that this $*$-autonomous structure yields a much richer type system for describing causal relationships between systems. Whereas compact closed categories have only one way to form joint systems, namely $\otimes$, $*$-autonomous categories have two: $\otimes$ and $\parr$. A simple, yet striking example of the difference in these two connectives in $\Caus[\mathcal C]$ comes from forming types of processes on a joint system:
\begin{align*}
  (\bm A \limp \bm A') \otimes (\bm B \limp \bm B') & \ \ \leftarrow\ \textrm{non-signalling processes} \\
  (\bm A \limp \bm A') \parr (\bm B \limp \bm B') & \ \ \leftarrow\ \textrm{all processes}
\end{align*}
Using the richer type system afforded by $*$-autonomous categories, we are able to give logical characterisations of many kinds of first- and higher-order causality constraints, and show, for example, when certain constraints imply others.

The paper is structured as follows. In section~\ref{sec:preliminaries}, we outline the basics of compact closed categories, string diagrams, $*$-autonomous categories, discarding, and (non-)signalling processes. In section~\ref{sec:precausal}, we introduce precausal categories and prove some basic properties, such as the \textit{no time-travel} theorem for precausal categories. In section~\ref{sec:causC}, we construct the higher-order causal category $\CausC$ and show that it is $*$-autonomous. In section~\ref{sec:first-order}, we develop various properties of first-order systems, most notably the coincidence of $\otimes$ and $\parr$. In section~\ref{sec:higher-order}, we demonstrate the aforementioned dichotomy of $\otimes$ and $\parr$ at the level of second-order sytems: namely that $\otimes$ can be used to construct a type of non-signalling processes whereas $\parr$ yields all processes. Furthermore, we show that one-way signalling processes can be captured in a HOCC using a generalisation of quantum combs, and processes satisfying arbitrary acyclic causal structures can be captured using pullbacks of combs. In section~\ref{sec:indef}, we show that HOCCs also enable us to capture bipartite and $n$-partite second-order causal processes. We give several examples which are known to exhibit indefinite causal structure, namely the OCB process from \cite{ViennaIndef}, the classical tripartite process from \cite{WolfClassicalMulti}, and an abstract version of the quantum switch defined in \cite{QSwitch}. Finally, we prove using just the structure of \CausC that the switch does not admit a causal ordering by reducing to no time-travel.

\textbf{Related work.} This work was inspired by \cite{PauloHierarchy}, which aims for a uniform description of higher-order \textit{quantum} operations in terms of generalised Choi operators. However, rather than relying on the linear structure of spaces of operators, we work purely in terms of the $*$-autonomous structure and the precausal axioms, which concern the compositional behaviour of discarding processes. The construction of \CausC is a variant of the \textit{double gluing construction} used in \cite{HylandGlue} to construct models of linear logic. In the language of that paper, our construction consists of building the `tight orthogonality category' induced by a focussed orthogonality on $\{1_I\}\subseteq \mathcal C(I,I)$, and then restricting to objects satisfying the flatness condition in Definition \ref{def:flat-and-closed}. 
Since it is $*$-autonomous, \CausC indeed gives a model of multiplicative linear logic, enabling us to enlist the aid of linear-logic based tools for proving theorems about causal types. We comment briefly on this in the conclusion. When fixing $\mathcal C = \MatRp$, the category $\Caus[\MatRp]$ is closely related to (the finite-dimensional subcategory of) probabilistic coherence spaces, introduced in~\cite{Quantic} and refined in~\cite{DanosCoherence}. The main difference, aside from allowing matrices over infinite sets of indices, is that the latter defines duals with respect to $[0,1] \subseteq \mathcal \MatRp(I,I) \cong \mathbb R_+$ rather than $\{1\}$, which allows one to capture sub-normalised probability distributions as well as normalised ones. 
Our construction indeed extends to incorporate sub-normalised processes, but understanding the structure of the resulting, larger category and its relationship to the one we define is a subject of future work. It is worth noting however, that the correspondence with coherence spaces does not extend to the quantum case. In particular, $\Caus[\CPM]$ is better behaved than the quantum coherence spaces defined by Gerard in~\cite{Quantic}, as it inherits the quantum tensor product---and hence the usual notions of positivity for states and complete positivity for morphisms---from the underlying category $\CPM$.

\textbf{Acknowledgements.} Both authors are grateful to Paulo Perinotti for his input and sharing a draft version of \cite{PauloHierarchy}. We would also like to thank \v{C}aslav Brukner, Matty Hoban, and Sam Staton for useful discussions. This work is supported by the ERC under the European Union's Seventh Framework Programme (FP7/2007-2013) / ERC grant n\textsuperscript{o} 320571.

\bigskip

\textit{This is an extended version of a conference paper with the same title~\cite{CatCausal}. While the overall structure is the same, it has been expanded with additional examples and explanation and provides two new technical contributions. The first is a proof that $\Rel$ is `weakly' precausal, in that it satisfies all of the precausal axioms except \Cf. The second and more substantial contribution is the pullback construction in Section~\ref{sec:pullback} which is used to capture all acyclic causal structures, rather than just linear ones as in~\cite{CatCausal}.}

\bigskip

\section{Preliminaries}\label{sec:preliminaries}

We work in the context of \textit{symmetric monoidal categories} (SMCs). An SMC consists of a collection of objects $\textrm{ob}(\mathcal C)$, for every pair of objects, $A, B \in \textrm{ob}(\mathcal C)$ a set $\mathcal C(A,B)$ of morphisms, associative sequential composition `$\circ$' with units $1_A$ for all $A \in \textrm{ob}(\mathcal C)$, associative (up to isomorphism) parallel composition `$\otimes$' for objects and morphisms with unit $I \in \textrm{ob}(\mathcal C)$, and swap maps $\sigma_{A,B} : A \otimes B \to B \otimes A$, satisfying the usual equations one would expect for composition and tensor product. For simplicity, we furthermore assume $\mathcal C$ is \textit{strict}, i.e.
\vspace{-0pt}
\[
 (A \otimes B) \otimes C = A \otimes (B \otimes C) \qquad\qquad A \otimes I = A = I \otimes A \vspace{-0pt}
\]
This is no loss of generality since every SMC is equivalent to a strict one. For details, \cite{MacLane} is a standard reference.

We wish to treat SMCs as theories of physical processes, hence we often refer to objects as \textit{systems} and morphisms as \textit{processes}. We will also extensively use \textit{string diagram} notation for SMCs, where systems are depicted as wires, processes as boxes, and:
\vspace{-0pt}
\begin{eqnarray*} 
g \circ f := \input{seq-comp.tikz} \qquad
   f \otimes g := \input{par-comp.tikz}\\
 1_A := \input{wire.tikz} \quad
   1_I := \emptydiag \quad
   \sigma_{A,B} := \input{swap.tikz} 
   \end{eqnarray*}
Note that diagrams should be read from bottom-to-top. A process $\rho : I \to A$ is called a \textit{state}, a process $\pi : A \to I$ is called an \textit{effect}, and $\lambda : I\to I$ is called a \textit{number} or \textit{scalar}. Depicted as string diagrams:\vspace{-0pt}
\[
\textit{state} := \pointmap{\rho}\quad
\textit{effect} := \copointmap{\pi}\quad
\textit{number} := \lambda
\]
Numbers in an SMC always form a commutative monoid with `multiplication' $\circ$ and unit the identity morphism $1_I$. We typically write $1_I$ simply as $1$.

\subsection{Compact closed categories and higher-order string diagrams} \ \\

\noindent We will begin with a category $\mathcal C$ and construct a new category $\CausC$ of higher-order causal processes. In order to make this construction, we first need a mechanism for expressing higher-order processes. \textit{Compact closure} provides such a mechanism that is convenient within the graphical language and already familiar within the literature on quantum channels, in the guise of the Choi-Jamio\l{}kowski isomorphism.

\begin{definition}\label{def:compact-closed}
  An SMC $\mathcal C$ is called \textit{compact closed} if every object $A$ has a \textit{dual} object $A^{*}$. That is, for every $A$ there exists morphisms $\eta_A: I \to A^{*} \otimes A$ and $\epsilon_A:A \otimes A^{*} \to I$, satisfying:
  \begin{equation*}
    (\epsilon_A \otimes 1_A) \circ (1_A \otimes \eta_A) = 1_A
    \qquad\qquad
    (1_{A^*} \otimes \epsilon_A) \circ (\eta_A \otimes 1_{A^*}) = 1_{A^*}
  \end{equation*}
\end{definition}

We refer to $\eta_A$ and $\epsilon_A$ as a cup and a cap, denoted graphically as \input{cup.tikz} and \input{cap.tikz}, respectively. In this notation, the equations in Definition \ref{def:compact-closed} become:
\[ \input{line_yank.tikz} \qquad\qquad \input{line_yank2.tikz} \]
It is always possible to choose cups and caps in such a way that the canonical isomorphisms $I^* \cong I$, $(A \otimes B)^* \cong A^* \otimes B^*$, and $A \cong A^{**}$ are all in fact equalities. We will assume this is the case throughout this paper.

Crucially, two morphisms in a free compact closed category are equal if and only if their string diagrams are the same. That is, if one diagram can be continuously deformed into the other while maintaining the connections between boxes. Hence, when we draw a string diagram, we mean \textit{any} composition of boxes via cups, caps, and swaps which yields the given diagram, up to deformation. See \cite{SelingerSurvey} for an overview of string diagram languages for monoidal categories.

Compact closed categories exhibit \textit{process-state duality}, 
that is, processes $f : A \to B$ are in 1-to-1 correspondence with states $\rho_f : I \to A^* \otimes B$:\vspace{-0pt}
\begin{equation}\label{eq:map-state}
  \unbig{\boxmap{f}} \ \ \mapsto\ \  \input{map_to_state.tikz}
\end{equation}

Hence, we treat everything as a `state' in $\mathcal C$ and write $f : X$ as shorthand for $f : I \to X$.
In this notation, states are of the form $\rho : A$, effects $\pi : A^*$, and general processes $f : A^* \otimes B$.
Furthermore, we won't require `output' wires to exit upward, and we allow irregularly-shaped boxes.
For example, we can write a process $w : A^* \otimes B \otimes C^* \otimes D$ using `comb' notation:\vspace{-0pt}
\begin{equation}\label{eq:comb-versions}
  \input{comb-versions.tikz}
\end{equation}
Note that we adopt the convention that an $A$-labelled `input' wire is of the same type as an $A^*$-labelled `output'.

While both the LHS and the RHS in equation~\eqref{eq:comb-versions} are notation for the same process $w$, the LHS is strongly suggestive of a second-order mapping from processes $B \to C$ to processes $A \to D$. Composition in this notation simply means applying the appropriate `cap' processes to plug wires together:
\[ \input{2_Comb-comp.tikz} \]

\begin{remark}
  Since oddly-shaped boxes don't uniquely fix any ordering of systems with respect to $\otimes$, we will often `name' each system by giving it a unique type and assume systems are permuted via $\sigma$-maps whenever necessary. This is a common practice e.g., in the quantum information literature.
\end{remark}

Our key examples of compact closed categories to which we will apply our construction to obtain categories of higher-order causal types, 
will be \MatRp and \CPM, which we introduce now.
They contain stochastic matrices and quantum channels, respectively.

\begin{example}
  The category \MatRp has as objects natural numbers. Morphisms $f : m \to n$ are $n \times m$ matrices. Composition is given by matrix multiplication: $(g \circ f)_i^j := \sum_k f_i^k g_k^j$, and the tensor is given by $m \otimes n := mn$ and $f \otimes g$ is the Kronecker product of matrices:\vspace{-0pt}
  \[ (f \otimes g)_{ij}^{kl} := f_i^k g_j^l \vspace{-0pt}\]
 Consequently, the tensor unit $I$ is the natural number $1$, so that states are column vectors $\rho : 1 \to n$, effects are row vectors $\pi : n \to 1$, and scalars are positive numbers $\lambda \in \mathbb R_+$.\\
 \MatRp is compact closed with $n = n^*$, where cups and caps are given by the Kronecker delta, $\delta_{ij}$, $\delta_{ij} = 1$ if $i = j$, $\delta_{ij} = 0$ otherwise. That is:
  \[ \eta^{ij} := \delta_{ij} =: \epsilon_{ij} \vspace{-0pt}\]
\end{example}


\begin{example}
The category \CPM has as objects the sets of linear operators, $\mathcal L(H)$, $\mathcal L(K), \ldots$, on finite dimensional complex Hilbert spaces $H,K,\ldots$ and morphisms $\Phi: \mathcal L(H) \to \mathcal L(K)$ are completely positive maps (i.e., positive maps which are still positive when tensored with other maps) with the usual composition: $\Phi \circ \Psi (a) = \Phi(\Psi(a))$.

The tensor product on objects satisfies $\mathcal L(H) \otimes \mathcal L(K) \cong \mathcal L(H \otimes K)$, and on morphisms it is given by linear extension of $\Phi \otimes \Psi (a \otimes b) = \Phi(a) \otimes \Psi(b)$.
Consequently, the tensor unit is $I = \mathcal L (\mathbb C) \cong \mathbb C$ and a state is a (completely) positive map $\rho:\mathbb C \to \mathcal L(H)$. Identifying such a map with its image on 1, we may say that a state is positive operator in $\mathcal L (H)$.
Effects are (un-normalised) quantum effects, i.e., (completely) positive maps $\pi: \mathcal L (H) \to \mathbb C$, and are of the form: $\pi(a) := \textrm{Tr}(P a)$, for some positive operator $P$. As with \MatRp, the scalars are $\mathbb R_+$. 

\CPM is also compact closed, with cups and caps given by the (un-normalised) Bell state:
\[ \eta = \ketbra{\Phi_0}{\Phi_0} \qquad \epsilon(\rho) = \textrm{Tr}(\rho \ketbra{\Phi_0}{\Phi_0})\]
where $\ket{\Phi_0} = \sum_i \ket{i} \otimes \ket{i}$. Consequentially, $\mathcal L(H)^* = \mathcal L(H)$ and equation \eqref{eq:map-state} gives the basis-dependent version, i.e., `Choi-style', of the Choi-Jamio\l{}kowski isomorphism \cite{DoingItWrong}.
\end{example}

\subsection{$*$-autonomous categories}\ \\

\noindent The biggest convenience of a compact closed structure is also its biggest drawback:
all higher-order structure collapses to first-order structure! For example, if we (temporarily) take $A \limp B := A^* \otimes B$ to be the object whose states correspond to processes from $A$ to $B$, we have:
\begin{align*}
  (A \limp B) \limp C 
  & := (A \limp B)^* \otimes C \\
  & := (A^* \otimes B)^* \otimes C \\
  & \cong A^{**} \otimes B^* \otimes C \\
  & \cong A \otimes B^* \otimes C \\
  & \cong B^* \otimes A \otimes C \\
  & =: B \limp A \otimes C
\end{align*}

As we will soon see, there is a pronounced difference between first-order causal processes which we introduce in the next section, and genuinely higher-order causal processes (see Chapter~\ref{sec:higher-order}). Hence, we would not expect such a collapse. If we look carefully at the calculation above, we see that things went wrong in the third step above. Because in any compact closed category, we always have $(A \otimes B)^* \cong A^* \otimes B^*$, we are able to distribute $(-)^*$ over $\otimes$. If we remove this condition, we obtain a new, weaker kind of category: 

\begin{definition}
  A \textit{$*$-autonomous category} is a symmetric monoidal category equipped with a full and faithful functor $(-)^* : \mathcal C^{\textrm{op}} \to \mathcal C$ such that, by letting:\vspace{-0pt}
  \begin{equation}\label{eq:limp-def}
    A \limp B := (A \otimes B^*)^*\vspace{-0pt}
  \end{equation}
  there exists a natural isomorphism:\vspace{-0pt}
  \begin{equation}\label{eq:closure}
    \mathcal C(A \otimes B, C) \cong \mathcal C(A, B \limp C)\vspace{-0pt}
  \end{equation}
\end{definition}

As the notation and the isomorphism \eqref{eq:closure} suggest, $A \limp B$ is the system whose states correspond to processes from $A$ to $B$.
Indeed, take $A = I$ in \eqref{eq:closure}.

Note that the definition we gave for $A \limp B$ differs from the one we gave for compact closed categories. Indeed any compact closed category is $*$-autonomous, where it additionally holds that:\vspace{-0pt}
\begin{equation}\label{eq:star-dist}
  A \otimes B \cong (A^* \otimes B^*)^*\vspace{-0pt}
\end{equation}
in which case:\vspace{-0pt}
\[ A \limp B := (A \otimes B^*)^* \cong A^* \otimes B^{**} \cong A^* \otimes B\vspace{-0pt} \]

However, in a $*$-autonomous category, the RHS of \eqref{eq:star-dist} is not $A \otimes B$, but something new, called the `par' of $A$ and $B$:\vspace{-0pt}
\begin{equation}\label{eq:par-dual}
  A \parr B := (A^* \otimes B^*)^*\vspace{-0pt}
\end{equation}
This new operation inherits its good behaviour from $\otimes$:\vspace{-0pt}
\[ A \parr (B \parr C) \cong (A \parr B) \parr C \qquad\qquad A \parr B \cong B \parr A \vspace{-0pt}\]

So a compact closed category is just a $*$-autonomous category where $\otimes = \parr$. However, this little tweak yields a much richer structure of higher-order maps.
We think of $A \otimes B$ as the joint state space of $A$ and $B$,
whereas $A \parr B$ is like taking the space of maps from $A^*$ to $B$.
For (first order) state spaces, these are basically the same, but as we go to higher order spaces, $A \parr B$ tends to be much bigger than $A \otimes B$.

We adopt the programmers' convention that $\otimes$ has precedence over $\limp$ and that $\limp$ associates to the right:
\[
A \otimes B \limp C := (A \otimes B) \limp C
\quad \quad A \limp B \limp C := A \limp (B \limp C)\vspace{-0pt}
\]
Either expression above represents the system whose states are processes with two inputs. Indeed \eqref{eq:closure} implies that $A \otimes B \limp C \cong A \limp B \limp C$. Since $\mathcal C$ is symmetric, we can re-arrange the inputs at will, i.e.,\vspace{-0pt}
\begin{equation}\label{eq:lolli-sym}
  A \limp B \limp C \cong B \limp A \limp C\vspace{-0pt}
\end{equation}

\subsection{Discarding, causality, and non-signalling}\label{sec:disc}\ \\

\noindent As noted in \cite{Cnonsig,CKpaperI,CRCaucat,Chiri2}, the crucial ingredient for defining causality is a preferred \textit{discarding} process $\disc_{\!\!A}$ from every system $A$ into $I$, which is compatible with the monoidal structure, in that 
\begin{equation}\label{eq:discarding-eqs}
  \discard_{\!\!A \otimes B} \ :=\  \discard_{\!\!A}\ \discard_{\!\!B}
\qquad
\discard_{\!\!I} \ :=\  1
\end{equation}

 Using this effect, we can define causality as follows:

\begin{definition}\label{def:causal}
  For systems $A$ and $B$ with discarding processes, a process $\Phi : A \to B$ is said to be \textit{causal} if:
  \begin{equation}\label{eq:causal}
    \input{causal.tikz}
  \end{equation}
\end{definition}

Intuitively, causality means that if we disregard the output of a process, it does not matter which process occurred. 

Hence \textit{causal states} 
 produce $1$ when discarded:\vspace{-0pt}
\[ \pointmap{\rho}\ \ \textit{causal} \ \ \iff \ \ \weight{\rho}\ =\ 1 \]
Since discarding the `output' of an effect $\pi : A \to I$ is the identity, there is a unique causal effect for any system, namely discarding itself:
\[ \copointmap{\pi}\ \ \textit{causal} \ \ \iff \ \ \copointmap{\pi}\ =\ \discard \]

\begin{example}\label{ex:mat-causality}
  For \MatRp, the discarding process is a row vector consisting entirely of $1$'s; it sends a state to the sum over its vector entries:\vspace{-0pt}
  \[ \discard\ =\ \left(\begin{matrix}
    1 & 1 & \cdots & 1
  \end{matrix}\right) \qquad \qquad  
  \weight{\rho} \ =\ \sum_i\ \rho^i \]
So, causality is precisely the statement that a vector of positive numbers sums to $1$, i.e., forms a probability distribution. Consequentially, the causality equation \eqref{eq:causal} for a process $\Phi$ states that each column of $\Phi$ must sum to $1$. That is, $\Phi$ is a stochastic map.
\end{example}

\begin{example}\label{ex:cpm-causality}
For \CPM, discard is the trace. Hence causal states are positive operators with unit trace a.k.a. density operators  and causal processes are trace-preserving completely positive maps, a.k.a. quantum channels.
\end{example}

Discarding not only allows us to express when a process is causal,
it also allows us to represent causal relationships between the systems involved.

For example
\begin{definition}\label{def:non-signalling}
A causal process $\Phi: A \otimes B \to A' \otimes B'$ is \emph{one-way signalling} with $\evt{A}$ before $\evt{B}$ (written $\evt{A} \preceq \evt{B}$ where $\evt A := (A, A')$ and $\evt B := (B, B')$)
if there exists a process $\Phi':A \to A'$ such that
\begin{equation}\label{eq:a-before-b}
  \input{NSBtoA.tikz}
\end{equation}
It is one-way signalling with $\evt{B} \preceq \evt{A}$ if there exists $\Phi''$ such that
\begin{equation}\label{eq:b-before-a}
  \input{NSAtoB.tikz}
\end{equation}
Such a process is called \emph{non-signalling} 
if it is both one-way signalling with $\evt{A} \preceq \evt{B}$ and $\evt{B} \preceq \evt{A}$.
\end{definition}

At first, this might seem like curious terminology, since non-signalling processes are a special case of one-way signalling processes. This is because the term `one-way signalling' does not imply that there \textit{is} communication from Alice to Bob, but rather that communication is \textit{only allowed} from Alice to Bob (and not from Bob to Alice). Contrast this with arbitrary causal processes $\Phi: A \otimes B \to A' \otimes B'$, which allow two-way communication in general.

Note that the expressions $\evt A \preceq \evt B$ and $\evt B \preceq \evt A$ in Definition~\ref{def:non-signalling} are not written in terms of individual inputs or outputs of $\Phi$, but rather input/output pairs $\evt A := (A, A')$ and $\evt B := (B, B')$. 
We call such an input/output pair from a given process an \textit{event}. 
If we think of an event in operational terms, it corresponds to a single agent giving an input to his or her black box (e.g., making a measurement choice) and then reading the output.

The definition of one-way signalling readily generalises from two events to $n$ events:

\begin{definition}
A process $\Phi$
is \emph{one-way signalling} with $\evt{A}_1 \preceq \ldots \preceq \evt{A}_{n-1} \preceq \evt{A}_n$ if
\ctikzfig{nfoldNS}
with $\Phi'$ one-way signalling with $\evt{A}_1 \preceq \ldots \preceq \evt{A}_{n-1}$.
\end{definition}

Intuitively, this captures that fact that $\Phi$ is consistent with the linear causal ordering of events $\evt A_i = (A_i, A_i')$. Following~\cite{KissingerIFF}, we can extend this from linear causal orderings to arbitrary (acyclic) causal orderings. For a process $\Phi : A_1 \otimes \ldots \otimes A_n \to A_1' \otimes \ldots \otimes A_n'$, let $\mathcal O$ be a partially ordered set whose elements are the pairs $\{ \evt A_i := (A_i, A_i') \}_{1 \leq i \leq n}$, which we call a \textit{causal ordering} for $\Phi$. Then, for any subset of events $\mathcal E \subseteq \mathcal O$ let $\past(\mathcal E)$ be the down-closure of $\mathcal E$ with respect to the ordering $\preceq$ of $\mathcal O$, i.e.,
\[ \past(\mathcal E) := \{ e \ |\ \exists e' \in \mathcal E\,.\, e \preceq e' \} \]
In particular, this set contains $\mathcal E$ itself. Furthermore, let $\pi_1(\mathcal E)$ and $\pi_2(\mathcal E)$ be given as the direct images of the first and second projections, namely all of the inputs in $\mathcal E$ and all the outputs in $\mathcal E$, respectively.

\begin{definition}\label{def:causal-consistency}
  A process $\Phi : A_1 \otimes \ldots \otimes A_n \to A_1' \otimes \ldots \otimes A_n'$ is consistent with a causal ordering $\mathcal O$ (written $\Phi \vDash \mathcal O$) if for all subsets $\mathcal E \subseteq \mathcal O$, the outputs of $\mathcal E$ only depend on the inputs of $\past(\mathcal E)$. That is, there exists another process $\Phi'$ such that:
  \begin{equation}\label{eq:causal-ord}
    \input{satisfies.tikz}
  \end{equation}
\end{definition}

Note that, for clarity, we have suppressed the symmetry morphisms (i.e., swap maps) re-ordering the inputs and outputs of $\Phi$ in \eqref{eq:causal-ord}. In particular, this equation should be understood to apply to all sets of inputs $\pi_1(\past(\mathcal E))$ and outputs $\pi_2(\mathcal E)$, not just the leftmost ones. 

Definition~\ref{def:causal-consistency} is most easily understood by means of an example. Consider the following process with 5 inputs and outputs:
\ctikzfig{big-phi}
and fix the following causal ordering on input/output pairs of $\Phi$:
\[
  \mathcal O := \left\{
   \ \input{causal-graph.tikz}\ 
  \right\} 
  \qquad \textrm{where} \qquad
  \begin{cases}
    \evt{A} := \text{\footnotesize $(A,A')$} \\
    \evt{B} := \text{\footnotesize $(B,B')$} \\
    \evt{C} := \text{\footnotesize $(C,C')$} \\
    \evt{D} := \text{\footnotesize $(D,D')$} \\
    \evt{E} := \text{\footnotesize $(E,E')$}
  \end{cases}
\]
(Note the ordering is depicted from bottom-to-top, e.g. $\evt{A} \preceq \evt{B}$.) Then, $\Phi \vDash \mathcal O$ if for all $\mathcal E \subseteq \mathcal O$, \eqref{eq:causal-ord} is satisfied. For example, taking $\mathcal E := \{ \evt{B} \}$, we have $\past(\{\evt{B}\}) = \{ \evt{A}, \evt{B}, \evt{C} \}$. So, condition~\eqref{eq:causal-ord} requires that there exists $\Phi' : A \otimes B \otimes C \to B'$ such that:
\ctikzfig{big-phi-eq-2}

Definition~\ref{def:causal-consistency} generalises the other (non)signalling conditions given before. For example, one-way signalling $\evt{A} \preceq \evt{B}$ and $\evt{B} \preceq \evt{A}$ can be stated respectively as:
\[
\input{bi-channel.tikz} \ \vDash \input{causal-graph-AB.tikz}
\qquad\textrm{and}\qquad
\input{bi-channel.tikz} \ \vDash \input{causal-graph-BA.tikz}
\]
whereas non-signalling is equivalent to the following:
\[
\input{bi-channel.tikz} \ \vDash \input{causal-graph-ns.tikz}
\]

There is a close connection between causality~\eqref{def:causal} and consistency with a causal ordering $\mathcal O$. Namely, any (acyclic) string diagram has an evident causal ordering on the inputs/outputs of each of the component boxes, where $\evt{A} \preceq \evt{B}$ if and only if there is a forward-directed path from the box with inputs/outputs $(A,A')$ to the box with inputs/outputs $(B,B')$. For example:
\[
\input{caus-circ.tikz}
\qquad\qquad \textrm{\large$\leadsto$} \qquad\qquad
\input{causal-graph-2.tikz}
\]
It was shown in \cite{KissingerIFF} that all acyclic diagrams in a category $\mathcal C$ respect their associated causal ordering if and only if all processes in $\mathcal C$ satisfy the causality equation~\eqref{eq:causal}.

In the coming sections, we will show that consistency with a causal ordering can be expressed within the logic of a higher-order causal category. As a consequence, higher-order constraints such as \textit{preserving} processes which are consistent with a causal ordering can also be expressed.

\section{Precausal categories}\label{sec:precausal}

Precausal categories give a universe of all processes, and provide enough structure for us to identify which of those processes satisfy first-order and higher-order causality constraints. They are defined as compact closed categories satisfying four extra axioms, which involve discarding and its transpose:
\begin{equation}\label{eq:maxmix-unnorm}
  \input{maxmix-disc.tikz}
\end{equation}
which we call the (unnormalised) \textit{uniform state}.

\begin{definition}\label{def:precausal}
  A \textit{precausal category} is a compact closed category $\mathcal C$ such that:
  \begin{itemize}[leftmargin=1cm] 
    \item[\Ca] $\mathcal C$ has discarding processes for every system, compatible with the monoidal structure as in \eqref{eq:discarding-eqs}.
    \item[\Cb] For every (non-zero) system $A$, the \textit{dimension} of $A$:
    \[ d_A \ :=\  \input{dim.tikz} \]
    is an invertible scalar.
    \item[\Cc] $\mathcal C$ has \textit{enough causal states}:
    \[ \left(\textrm{\footnotesize $\forall \rho\ \textit{causal}\ .\ $}\input{causal-agree.tikz}\right)
     \implies\ \  \input{causal-agree1.tikz} \]
    \item[\Cd] \textit{Second-order causal} processes factorise:
    \[
    \quant{\forall \Phi\ \textit{causal}}{\input{soc.tikz}}
    \implies
    \quant{\exists \Phi_1,\Phi_2\ \textit{causal}}{\input{soc-factor.tikz}}
    \]
  \end{itemize}
\end{definition}

\Ca enables one to talk about causal processes within $\mathcal C$. \Cb enables us to renormalise certain processes to produce causal ones. For example, every non-zero system has at least one causal state, called the \textit{uniform state}. It is obtained by normalising \eqref{eq:maxmix-unnorm}:\vspace{-0pt}
\ctikzfig{uniform-def}
Then:
\begin{equation}\label{eq:uniform-causal}
	\input{uniform-causal.tikz}
\end{equation}

\begin{remark}
Note that by `non-zero' in \Cb, we mean `not a zero object'. For our examples, it will be convenient to allow $\mathcal C$ to have a zero object (e.g. the natural number $0$ in \MatRp and the zero-dimensional Hilbert space in \CPM), in which case $d_0 = 0$ is certainly not invertible.
\end{remark}

\Cc says that processes are characterised by their behaviour on causal states. In the following proposition, we will show that \Cc implies that it suffices to consider only product states to identify a process.

\begin{proposition}\label{prop:enough-product-states}
For any compact closed category $\mathcal C$, \Cc is equivalent to:
\begin{equation}\label{eq:enough-product-states}
  \quant{\forall \rho_1, \rho_2 \ \textit{causal}}
{\input{enough-ls4.tikz}}
\ \implies\ 
\input{remove-one1.tikz}
\end{equation}
\end{proposition}

\begin{proof}
\Cc follows from \eqref{eq:enough-product-states} by taking one of the two systems involved to be trivial. Conversely, assume the premise of \eqref{eq:enough-product-states}. Applying \Cc one time yields:
\ctikzfig{remove-one}
for all causal states $\rho_2$. Bending the wire yields:
\ctikzfig{remove-one-bend}
Hence we can apply \Cc a second time. Bending the wire back down gives the result.
\end{proof}

\Cd is perhaps the least transparent. It says that the only mappings from causal processes to causal processes are `circuits with holes', i.e., those mappings which arise from plugging a causal process into a larger circuit of causal processes. This can equivalently be split into two smaller pieces, which will be helpful in showing \MatRp and \CPM form examples of precausal categories.

\begin{proposition}\label{4to45}
  For a compact closed category $\mathcal C$ satisfying \Ca, \Cb, and \Cc, condition \Cd is equivalent to the following two conditions:
  \begin{itemize}[leftmargin=1cm] 
    \item[\Ce] Causal \textit{one-way signalling} processes factorise:
      \[ \quant{\exists\ \Phi'\ \textit{causal}}{\input{one-way-sig.tikz}}
         \implies
         \quant{\exists\ \Phi_1,\Phi_2\ \textit{causal}}{\input{one-way-sig2.tikz}}
      \]
      \item[\Cf] For all $w : A \otimes B^*$:
      \[ \quant{\forall \Phi \ \textit{causal}}
         {\input{causal-eff.tikz}\ =\ 1}
         \implies
         \quant{\exists \rho \ \textit{causal}}
         {\input{causal-eff-splits.tikz}} \]
  \end{itemize}
\end{proposition}

\begin{proof}
Assume \Cd holds, and let $\Phi$ satisfy the premise of \Ce. Then, for causal $\Psi$, we have:
\ctikzfig{SOC-pf}
Hence, by \Cd there exist causal $\Phi_1', \Phi_2$ such that:
\ctikzfig{SOC-pf2}
Deforming then gives the factorisation in \Ce.

To get to \Cf from \Cd, we take the input and output systems of $w$ trivial, which gives:
\ctikzfig{soc-factor_trivial_systems}
where the second `$=$' follows from the fact that discarding is the unique causal effect.
Conversely, suppose \Ce and \Cf hold. Let $w$ be a second order causal process. Then, for any causal $\rho$ and $\Phi$:
\[ \input{causal-eff-splits1.tikz}\ =\ 1 \]
So, by \Cf:
\ctikzfig{causal-eff-splits2}
Then, by \Cc:
\ctikzfig{causal-eff-splits3}
Hence $w$ satisfies the premise of \Ce, up to diagram deformation, where $\Phi := w$ and $\Phi' := w'$, which implies that it factors as in \Cd.
\end{proof}

\begin{remark}
  Whereas axiom \Cf is just the restriction of \Cd to a special case, \Ce is an abstract version of a familiar result about quantum channels. In the literature, one-way signalling is sometimes called \textit{semi-causal}, whereas the factorisation of $\Phi$ into processes $\Phi_1, \Phi_2$ in \Ce is referred to a \textit{semi-localisable}. Whereas it is immediately clear that semi-localisable implies semi-causal, the converse, i.e. condition \Ce, is non-trivial. It was conjectured for quantum channels by DiVincenzo and independently by Beckman et al~\cite{BGNP}, and proven by Eggeling, Schlingemann, and Werner in~\cite{Eggeling}. In the case of classical probabilistic processes, \Ce can be seen as an instance of the product rule, applied to conditional probabilities of the form $P(A'B'|AB)$ satisfying the conditional independence $P(A'|AB) = P(A'|A)$. See Appendix~\ref{sec:mat-cpm-precausal} for details.
\end{remark}

\begin{lemma}\label{lem:splitting}
  For any $w : A \otimes B^*$:
  \[
  \lquant{\exists \rho}{\input{causal-eff-splits.tikz}}
  \ \iff\ 
  \left(\ \input{SOC_splits_lemma-long.tikz}\ \right)
  \]
\end{lemma}
\begin{proof}
  $(\Leftarrow)$ immediately follows from diagram deformation. For $(\Rightarrow)$, assume the leftmost equation above. Then we can obtain an expression for $\rho$ by plugging the uniform state into $w$ and applying causality of the uniform state:
  \ctikzfig{rho-as-comb}
  Substituting this expression back in for $\rho$ yields:
  \[
    \input{rho-as-comb1.tikz}
    \tag*{\qedhere}
  \]
\end{proof}

\begin{example}
  \MatRp and \CPM are both precausal categories. See appendix for proofs of conditions \Ca-\Cd.
\end{example}

\begin{example}\label{ex:rel-weakly-precausal}
  \Rel, the category whose objects are sets and whose morphisms $R : X \to Y$ are relations $R \subseteq X \times Y$ with $R \otimes S := R \times S$, is \textit{not} a pre-causal category. Nevertheless, it satifies axioms \Ca-\Cc and \Ce, so it could be called a \textit{weakly pre-causal category}. See appendix for proofs and a counter-example of \Cf.
\end{example}

We now show that \Ce implies a general $n$-fold version of itself:

\begin{proposition}\label{w factorization}
Let $\Phi$ be one-way signalling with $\evt A_1 \preceq \ldots \preceq \evt A_n$, then there exists $\Phi_1,\ldots, \Phi_n$ such that
\ctikzfig{seq-of-channels}
\end{proposition}

\begin{proof}
For $n=2$, this is just \Ce.
Suppose the proposition is true for $n-1$.
Then because
\ctikzfig{n-1discards}
for some one way signalling process $\Phi'$ with $\evt A_1 \preceq \ldots \preceq \evt A_{n-1}$,
we have by \Ce that there exists $\Phi'_{n-1}$ and $\Phi_n$ such that
\ctikzfig{n-1and1}
It follows that $\Phi'$ equals $\Phi'_{n-1}$ with the $C$ system discarded,
so that $\Phi'_{n-1}: A_1 \otimes \ldots \otimes A_{n-1} \to A'_1 \otimes \ldots \otimes (A'_{n-1} \otimes C)$ is again one-way signalling.
By assumption $\Phi'_{n-1}$ now factors and hence so does $\Phi$.
\end{proof}


While, as we shall soon see, precausal categories give us a source of processes exhibiting many varieties of definite and indefinite causal structure, the axioms rule out certain, paradoxical causal structures. To see this, we state our first no-go result for a precausal category $\mathcal C$.

\begin{theorem}[No time-travel]\label{thm:no-time-travel}
  No non-trivial system $A$ in a precausal category $\mathcal C$ admits \textit{time travel}. 
  That is, if there exist systems $B$ and $C$ such that for all processes $\Phi$ we have:
  \begin{equation}\label{eq:time-travel}
    \input{time-phi.tikz}\  \textit{causal} \qquad \implies \qquad \input{time-travel.tikz}\ \textit{causal}\vspace{-0pt}
  \end{equation}
  then $A \cong I$.
\end{theorem}

\begin{proof}
  For any causal process $\Psi : A \to A$, we can define:\vspace{-0pt}
  \ctikzfig{time-phi1}\vspace{-0pt}
  which is also a causal process. Then implication \eqref{eq:time-travel} gives:
  \[ \input{time-travel1.tikz} \vspace{-0pt}\]
  Applying $\Cf$, we have:
  \[ \input{time-travel3.tikz} \qquad \implies \qquad \input{time-travel4.tikz} \]
  for some causal state $\rho : I \to A$. That is, $\rho \circ \disc = 1_A$, and by definition of causality for $\rho$, $\disc \circ \rho = 1_I$.
\end{proof}

Note that a special case of Theorem \ref{thm:no-time-travel} implies that if for all causal processes $\Phi : A \to A$ we have
\ctikzfig{Trace-timetravel}
then $A \cong I$.

\section{Constructing \CausC}\label{sec:causC}

We will now describe our main construction, the category $\CausC$ of higher-order causal processes for a precausal category $\mathcal C$. To motivate this construction, we begin by looking at the properties of the set of causal states for some system $A$:
\begin{equation}\label{eq:all-causal-states}
  \left\{ \rho : A \ \bigg|\ \weight{\rho} \ =\ 1 \right\} \subseteq \mathcal C(I, A)
\end{equation}
In the classical and quantum cases, these form convex subsets of real vector spaces (probability distributions in $\mathbb R^n$ and density matrices in $\textrm{SA}_n$, the self adjoint $n \times n$ operators, respectively). We would like to recapture the fact that this set is suitably closed without referring to convexity, so we appeal to duals instead.

\begin{definition}
For any set of states $c \subseteq \mathcal C(I, A)$, we can define the dual set $c^* \subseteq \mathcal C(I, A^*)$ as follows:
\[ c^*\ :=\ \left\{ \pi : A^* \ \bigg|\ \forall \rho \in c\ .\ \unbig{\pointbraket{\rho}{\pi}} \ =\ 1 \right\} \]
\end{definition}

Taking the dual again, we get back to a set of states in $A$. It immediately follows from the definition that $c \subseteq c^{**}$. For the set of all causal states \eqref{eq:all-causal-states}, we see that furthermore $c = c^{**}$. Assuming this property, which we call \textit{closure}, for all objects will play an important role in obtaining $*$-autonomous structure in $\CausC$.

However, this property alone only refers to the compact closed structure of $\mathcal C$ and not the precausal structure. By definition, the discarding process is contained in the dual $c^*$ of the set of all causal states. We also showed in equation~\eqref{eq:uniform-causal} that the transpose of discarding is contained in the set of causal states, up to a scalar (namely $d_A^{-1}$). Generalising this to a property that is symmetric in the roles of $c$ and $c^*$, we obtain the notion of \textit{flatness}.

\begin{definition}\label{def:flat-and-closed}
  A set of states $c \subseteq \mathcal C(I, A)$ is \textit{closed} if $c = c^{**}$ and \textit{flat} if there exist invertible scalars $\lambda, \mu$ such that:\vspace{-0pt}
  \[ \lambda\ \maxmix\ \in c \qquad\qquad \mu\ \discard \in c^* \]
\end{definition}

As we have already noted, the set of all causal states is closed and flat. Many other sets turn out to be closed and flat, including the sets of causal processes and higher-order generalisations thereof. We now use these two properties to define a category which incorporates all of these higher-order types.

\begin{definition}\label{def:causC}
  For a precausal category $\mathcal C$, the category $\CausC$ has as objects pairs:\vspace{-0pt}
  \[\bm A := (A, c_{\bm A} \subseteq \mathcal C(I, A))\vspace{-0pt} \]
  where $c_{\bm A}$ is closed and flat. A morphism $f : \bm A \to \bm B$ is a morphism $f : A \to B$ in $\mathcal C$ such that:\vspace{-0pt}
  \begin{equation}\label{eq:caus-condition}
    \rho \in c_{\bm A}\ \implies\ f \circ \rho \in c_{\bm B}\vspace{-0pt}
  \end{equation}
\end{definition}

We refer to categories of the form $\Caus[\mathcal C]$ for some precausal category $\mathcal C$ as \textit{higher-order causal categories} (HOCCs).

While condition~\eqref{eq:caus-condition} is given in terms of states, closure allows us to equivalently present it in terms of effects or numbers:

\begin{proposition}\label{prop:morphism}
  For objects $\bm A, \bm B$ in $\CausC$ and a morphism $f : A \to B$ in $\mathcal C$, the following are equivalent:
  \begin{itemize}[leftmargin=1cm]
    \item[(i)] $\rho \in c_{\bm A} \implies f \circ \rho \in c_{\bm B}$
    \item[(ii)] $\pi \in c_{\bm B}^* \implies \pi \circ f \in c_{\bm A}^*$
    \item[(iii)] $\rho \in c_{\bm A}, \pi \in c_{\bm B}^* \implies \pi \circ f \circ \rho = 1$
  \end{itemize}
\end{proposition}

\begin{proof}
  (i) $\Rightarrow$ (ii) and (ii) $\Rightarrow$ (iii) follow immediately from the definition of $(-)^*$, so assume (iii) and take any $\rho \in c_{\bm A}$. Then for all $\pi \in c_{\bm B}^*$, $\pi \circ (f \circ \rho) = 1$. Hence $f \circ \rho \in c_{\bm B}^{**} = c_{\bm B}$.
\end{proof}

Since a set of states is closed when $c = c^{**}$, it is natural to ask if $(-)^{**}$ forms a closure operation, namely if it is idempotent. This is an immediate result of the following:

\begin{lemma}\label{lem:triple-dual}
 For any set of states $c$ we have $c^{*} = c^{***}$.
\end{lemma}

\begin{proof}
First, note that:
\[ c \subseteq d \ \implies \ d^* \subseteq c^*. \]
Applying this to $c \subseteq c^{**}$ yields $c^{***} \subseteq c^*$.
But then, it is already the case that $c^*$ is contained in $c^{***}$, 
so $c^{***} = c^*$.
\end{proof}

We will now show that $\CausC$ has the structure of a $*$-autonomous category. To do this, we will first define the tensor $\bm A \otimes \bm B$. For the sets of states $c_{\bm A}$ and $c_{\bm B}$, we denote the set of all product states as follows:\vspace{-0pt}
\[ c_{\bm A} \otimes c_{\bm B} \ :=\ \{ \rho_1 \otimes \rho_2 \ |\ \rho_1 \in c_{\bm A},\ \rho_2 \in c_{\bm B} \}\vspace{-0pt} \]
Then, $c_{\bm A \otimes \bm B}$ is the closure of the set of all product states:\vspace{-0pt}
\[ c_{\bm A \otimes \bm B} := (c_{\bm A} \otimes c_{\bm B})^{**}\vspace{-0pt} \]

\begin{lemma}\label{lem:tensor-effects}
  For any effect $\pi : A^* \otimes B^*$ in $\mathcal C$:
  \begin{equation}\label{eq:tensor-effects}
    \quant{\forall\ \rho \in c_{\bm A \otimes \bm B}}
    {\input{bi-braket1.tikz}\ =\ 1}
    \ \ \iff\ \ 
    \quant{\forall\ \rho_1 \in c_{\bm A}, \rho_2 \in c_{\bm B}}
    {\input{bi-braket2.tikz}\ =\ 1}
  \end{equation}
\end{lemma}

\begin{proof}
  The LHS of \eqref{eq:tensor-effects} states that\vspace{-0pt}
  \[ \pi \in c_{\bm A \otimes \bm B}^* := ((c_{\bm A} \otimes c_{\bm B})^{**})^* = (c_{\bm A} \otimes c_{\bm B})^{***} \vspace{-0pt}\]
  whereas the RHS states that $\pi \in (c_{\bm A} \otimes c_{\bm B})^{*}$. Hence, \eqref{eq:tensor-effects} follows from Lemma \ref{lem:triple-dual}.
\end{proof}

\begin{theorem}\label{CausCisSMC}
  $\CausC$ is an SMC, with tensor given by:\vspace{-0pt}
  \[ \bm A \otimes \bm B := (A \otimes B, c_{\bm A \otimes \bm B})\vspace{-0pt} \]
  and tensor unit $\bm I := (I, \{1\})$.
\end{theorem}
\begin{proof}
The proof is in the appendix.
\end{proof}

Now define objects $\bm A^* := (A^*, c_{\bm A^*})$ in the obvious way, by letting $c_{\bm A^*} := c_{\bm A}^*$. Then

\begin{lemma}\label{lem:star}
  The transposition functor $(-)^* : \mathcal C^{\textrm{op}} \to \mathcal C$:
  \begin{equation}\label{eq:trans-functor}
    A\mapsto A^* \qquad\qquad \input{transpose.tikz}
  \end{equation}
  lifts to a full and faithful functor $(-)^* : \CausC^{\textrm{op}} \to \CausC$, where $\bm A^* := (A^*, c_{\bm A^*})$.
\end{lemma}
\begin{proof}
The main part is showing that $f^*$ is again a morphism, but this follows from the definition of the star on sets of states.
A full proof is in the appendix.
\end{proof}

We now have enough structure to define $\bm A \limp \bm B := (\bm A \otimes \bm B^*)^*$. However, it is enlightening to give an explicit characterisation of the set $c_{\bm A \limp \bm B}$. This will be no surprise:

\begin{lemma}\label{lem:lolli-explicit}
  For objects $\bm A, \bm B \in \CausC$:\vspace{-0pt}
  \[ c_{\bm A \limp \bm B} = \left\{ f : A^* \otimes B\ \bigg|\ \forall \rho \in c_{\bm A}, \pi \in c_{\bm B}^*\ .\ \input{dual-f1.tikz}\ =\ 1 \right\} \]
\end{lemma}

\begin{proof}
  This follows by simplifying:\vspace{-0pt}
  \[
    c_{(\bm A \otimes \bm B^*)^*}
       = c_{(\bm A \otimes \bm B^*)}^* 
       = (c_{\bm A} \otimes c_{\bm B^*})^{***} 
       = (c_{\bm A} \otimes c_{\bm B}^*)^{*}
   \]
  and noting that $f \in (c_{\bm A} \otimes c_{\bm B}^*)^{*}$ is precisely the statement given in the lemma.
\end{proof}

\begin{theorem}\label{thm:causC}
  For any precausal category $\mathcal C$, \CausC is a $*$-autonomous category where $\bm I = \bm I^*$.
\end{theorem}
\begin{proof}
Since compact closed categories already admit an interpretation for $\limp$ satisfying \eqref{eq:closure}, it suffices to show that this isomorphism lifts to \CausC. This follows from the application of Lemma~\ref{lem:tensor-effects}. The complete proof is in the appendix.
\end{proof}

Since \CausC is $*$-autonomous, we can also define the `par' of two systems $\bm A \parr \bm B$. Since $\parr$ is related to $\limp$ via $\bm A \parr \bm B \cong \bm A^* \limp \bm B$, Lemma \ref{lem:lolli-explicit} also yields an explicit form for $c_{\bm A \parr \bm B}$ by replacing $\bm A$ with $\bm A^*$:\vspace{-0pt}
\[
c_{\bm A \parr \bm B} =
\left\{ \rho : A \otimes B\ \bigg|\ \textrm{\footnotesize$\forall \pi \in c_{\bm A}^*, \xi \in c_{\bm B}^*\ .\ $}\input{par.tikz} = 1 \right\} \]
Note that the process $f : A^* \otimes B$ has become a bipartite state $\rho : A \otimes B$. That is, the states $\rho : \bm A \parr \bm B$ are states which are normalised for all product effects. Symbolically, the two monoidal products are defined as follows:\vspace{-0pt}
\[
c_{\bm A \otimes \bm B} = (c_{\bm A} \otimes c_{\bm B})^{**}
\qquad
c_{\bm A \parr \bm B} = (c_{\bm A}^* \otimes c_{\bm B}^*)^*\vspace{-0pt}
\]
One can easily check that $(c_{\bm A}^* \otimes c_{\bm B}^*) \subseteq (c_{\bm A} \otimes c_{\bm B})^*$. Thus, since $(-)^*$ reverses subset inclusions, that
$c_{\bm A \otimes \bm B} \subseteq c_{\bm A \parr \bm B}$.
Consequently, the identity $1_{A \otimes B}$ in $\mathcal C$ lifts to a canonical embedding $\bm A \otimes \bm B \hookrightarrow \bm A \parr \bm B$ in \CausC. This agrees with the intuition given in Section~\ref{sec:precausal} that $\bm A \parr \bm B$ is the `larger' of the two ways to combine $A$ and $B$ into a joint system.

\begin{remark}
  A $*$-autonomous category with coherent isomorphism $I \cong I^*$, such as \CausC, is also called an ISOMIX category~\cite{ISOMIX}. This innocent-looking extra condition actually gives a great deal more structure. For instance, even though we showed it concretely, the existence of a canonical morphism $A \otimes B \to A \parr B$ is implied purely from this extra structure.
\end{remark}

Rather than thinking of \CausC as a totally new category constructed from $\mathcal C$, it is useful to think of it as endowing the processes in $\mathcal C$ with a much richer type system. As in the compact closed case, it suffices to consider only processes out of $\bm I$ and we use $\rho : \bm X$ as shorthand for $\rho : \bm I \to \bm X$. However, unlike before, we will often use a statement of the form $\rho : \bm X$ as a \textit{proposition} about a state $\rho \in \mathcal C(I, X)$.

\begin{proposition}
  For a system $\bm X = (X, c_{\bm X})$ in \CausC and a state $\rho \in \mathcal C(I, X)$, $\rho : \bm X$ if and only if $\rho \in c_{\bm X}$.
\end{proposition}

\begin{proof}
  Since $1$ is the unique state in $c_{\bm I}$, the result follows immediately from the definition of morphism in \CausC:\vspace{-0pt}
  \[
    \rho : \bm X \ \iff\  \rho \circ 1 \in c_{\bm X} \ \iff\ \rho \in c_{\bm X}
    \tag*{\qedhere}
  \]
\end{proof}

From now on, we will use $\rho : \bm X$ and $\rho \in c_{\bm X}$ interchangeably without further comment. We will also mix the graphical notation with the type-theoretic. So, for instance, if we write:\vspace{-0pt}
\[
\input{composition.tikz}\ \ :\ \bm A \limp (\bm B \limp \bm C) \limp \bm D\vspace{-0pt}
\]
this should be interpreted as a morphism in $\mathcal C(I, A^* \otimes B \otimes C^* \otimes D)$, along with the assertion that this morphism has type $\bm A \limp (\bm B \limp \bm C) \limp \bm D$. In particular, diagrams always depict $\mathcal C$-morphisms, as opposed to \CausC-morphisms, so there is no ambiguity about whether parallel composition means $\otimes$ or $\parr$.

Furthermore, if we state that two types are isomorphic without giving the isomorphism explicitly, it should be understood that the underlying isomorphism in $\mathcal C$ is just the identity, up to a possible permutation of systems. In particular, $\bm X \cong \bm Y$ implies that $\rho : \bm X$ if and only if $\rho : \bm Y$.

\section{First order systems}\label{sec:first-order}

For any precausal category $\mathcal C$, we can always form the SMC of (first-order) causal processes $\mathcal C_c$ by restricting just to those processes satisfying the causality equation~\eqref{eq:causal}. Since \CausC is supposed to contain first and higher-order causal processes, one would naturally expect $\mathcal C_c$ to embed in \CausC.

\begin{definition}
A system $\bm A = (A, c_{\bm A})$ in $\CausC$ is called \emph{first order} if it is of the form $(A, \{ \disc_{\!\!A} \}^*)$.
\end{definition}

Note that $\{ \disc_{\!\!A} \}^*$ is precisely the set \eqref{eq:all-causal-states} of causal states of type $A$. Clearly this set is flat and closed. Indeed, this was the motivation for these conditions in the first place. Now, we show that the processes between first-order systems in \CausC are exactly as expected.

\begin{proposition}\label{prop:causal-char}
Let $\bm A, \bm B$ be first-order systems. Then $f$ is a morphism from $\bm A$ to $\bm B$ if and only if it is causal.
\end{proposition}

\begin{proof}
  We first compute $c_{\bm A}^*$ for a first-order system. Suppose $\pi \in c_{\bm A}^*$. Then for all causal states $\rho$, $\pi \circ \rho = 1 = \disc_{\!\!A} \circ \rho$, so by \Cc $\pi = \disc_{\!\!A}$. Hence $c_{\bm A}^* = \{ \disc_{\!\!A} \}$.

  Now, by Proposition~\ref{prop:morphism}, $\Phi \in \mathcal C(A, B)$ is a morphism from $\bm A$ to $\bm B$ if and only if for every $\pi \in c_{\bm B}^*$, $\pi \circ \Phi \in c_{\bm A}^*$. Since both of these sets of effects only contain discarding, this reduces to the causality equation \eqref{eq:causal}.
\end{proof}

Furthermore, first-order systems are closed under $\otimes$.

\begin{proposition}\label{prop:fo-tensor}
  For first order systems $\bm A$, $\bm B$, $\bm A \otimes \bm B$ is also a first-order system, given by:\vspace{-0pt}
  \[ \bm A \otimes \bm B = (A \otimes B, \{ \disc_{\!\!A} \disc_{\!\!B} \}^*)\vspace{-0pt} \]
\end{proposition}

\begin{proof}
  It suffices to show that $c_{\bm A \otimes \bm B}^* = \{ \disc_{\!\!A} \disc_{\!\!B} \}$. Let $\pi \in c_{\bm A \otimes \bm B}^*$, then for all causal states $\rho_1 \in c_{\bm A}, \rho_2 \in c_{\bm B}$:
  \[ \input{bi-braket2.tikz}\ =\ 1 \]
  Hence, by Proposition \ref{prop:enough-product-states}, $\pi = \disc_{\!\!A} \disc_{\!\!B}$.
\end{proof}

\begin{corollary}\label{cor:embedding}
  There exists a full, faithful, monoidal embedding of the category $\mathcal C_c$ of causal processes into \CausC via:\vspace{-0pt}
  \[ A \mapsto (A,\{\disc_{\!\!A}\}^*) \qquad\qquad f \mapsto f \vspace{-0pt}\]
\end{corollary}

Hence, the full sub-category of first-order systems and processes behaves as expected; it is equivalent to $\mathcal C_c$. Perhaps a more surprising corollary to Proposition \ref{prop:fo-tensor} is the following.

\begin{corollary}\label{cor:fo-tensor-is-par}
Let $\bm A$ and $\bm B$ be first order systems, then:\vspace{-0pt}
\[ \bm A \otimes \bm B \cong \bm A \parr \bm B \vspace{-0pt}\]
\end{corollary}

\begin{proof}
  $c_{\bm A \parr \bm B} := (c_{\bm A}^* \otimes c_{\bm B}^*)^* = \{ \disc_{\!\!A} \disc_{\!\!B} \}^* = c_{\bm A \otimes \bm B}$.
\end{proof}

So for first-order systems, there is really only one way to form the `joint system'. However, we will now see that for higher-order systems, this is very much not the case.





\section{Higher-order systems with definite causal order}\label{sec:higher-order}

While it is important that the category of causal processes embeds fully and faithfully in \CausC when one restricts to first-order systems, the chief interest of \CausC are its higher-order systems.
The goal of this section is to show that certain collections of maps fit nicely within the developed type theory.

The first non-trivial second-order system that it is natural to consider is $\bm A \limp \bm B$ for first-order systems $\bm A$, $\bm B$. The isomorphism \eqref{eq:closure} for $*$-autonomous categories restricts to:\vspace{-0pt}
\[ \CausC(\bm A, \bm B) \cong \CausC(\bm I, \bm A \limp \bm B) \vspace{-0pt}\]
so `states' $\Phi : \bm A \limp \bm B$ are in bijective correspondence with morphisms from $\bm A$ to $\bm B$ in \CausC. That is, they are precisely the causal processes from $\bm A$ to $\bm B$.

Now, starting from first-order systems $\bm A, \bm A', \bm B, \bm B'$, we have two ways to form the `joint system' from $\bm A \limp \bm A'$ and $\bm B \limp \bm B'$, via $\otimes$ and $\parr$. Before we characterise these systems, we examine the dual of a second-order system.

If we take a process $w: (\bm A \limp \bm B)^*$ in the dual system, we know from \Cd that that $w$ must split into two pieces, a causal state and the discard map. In fact, using flatness, we can strengthen this condition by only requiring $\bm B$ to be first-order.

\begin{lemma}\label{lem:ho-splitting}
  For any system $\bm X$, first-order system $\bm B$, and process $w : (\bm X \limp \bm B)^*$, there exists $\rho : \bm X$ such that:\vspace{-0pt}
  \begin{equation}\label{eq:ho-splitting}
    \input{ho-causal-eff-splits-types.tikz}
  \end{equation}
\end{lemma}
\begin{proof}
  Since $\bm B$ is first-order, $c_{\bm B}^* = \{\disc_{\!\!B}\}$ and by flatness, for some $\mu$, $\mu \disc_{\!\!X} \in c^*_{\bm X}$. Hence, for any (first-order) causal process $\Phi : X \to B$, we have $\disc_{\!\!B} \circ \mu \Phi = \mu\ \disc_{\!X}$. Hence $\mu \Phi$ is of type $\bm X \to \bm B$. It follows by definition of $(\bm X \limp \bm B)^*$ that:\vspace{-0pt}
  \ctikzfig{ho-causal-eff}
  That is, $\mu w$ sends every causal process $\Phi$ to $1$, so by $\Cf$:\vspace{-0pt}
  \[ \left(\input{ho-causal-eff-splits.tikz}\right)
     \quad\implies\quad
     \left(\input{ho-causal-eff-splits1.tikz}\right)\vspace{-0pt}
  \]
  It is then straightforward to show that $\rho := \mu^{-1} \rho'$ is a state of $\bm X$, e.g., by plugging the uniform state for system $\bm B$ into both sides of the equation above.
\end{proof}

Equivalently, by Lemma \ref{lem:splitting}, $w : (\bm X \limp \bm A)^*$ implies:
\begin{equation}\label{eq:ho-splitting-explicit}
  \input{SOC_splits_lemma.tikz}
\end{equation}




We now characterize the two ways to combine the spaces $\bm A \limp \bm A'$ and $\bm B \limp \bm B'$:

\begin{theorem}\label{type of non-signalling}
For first-order systems $\bm A, \bm A', \bm B, \bm B'$, a process
\[
\input{bi-phi.tikz}
\]
in $\mathcal C$ has type $(\bm A \limp \bm A')\otimes (\bm B \limp \bm B')$ in \CausC if and only if it is causal and non-signalling.
\end{theorem}

\begin{proof}
First assume that $\Phi : (\bm A \limp \bm A')\otimes (\bm B \limp \bm B')$. Then, since discarding $B'$ is causal, we can regard it as a morphism $\disc_{\!\!B'} : \bm B' \to \bm I$. Hence by functoriality of $\otimes$ and $\limp$, we have:\vspace{-0pt}
\[ \input{disc-B.tikz}\ \ :\ (\bm A \limp \bm A')\otimes (\bm B \limp \bm I) \vspace{-0pt}\]
Then we can transform to an equivalent type as follows:\vspace{-0pt}
\begin{align*}
  (\bm A \limp \bm A') \otimes (\bm B \limp \bm I)
  & \cong (\bm A \limp \bm A') \otimes \bm B^* \\
  & \cong ((\bm A \limp \bm A')^* \parr \bm B)^* \\
  & \cong ((\bm A \limp \bm A') \limp \bm B)^*\vspace{-0pt}
\end{align*}
Hence, by Lemma \ref{lem:ho-splitting}, $\Phi$ splits as $\disc_{\!\!B} : \bm B$ and $\Phi' : \bm A \limp \bm A'$. This gives exactly the first non-signalling equation:
\ctikzfig{NSBtoA}
The second equation is shown similarly, by plugging in $\disc_{\!\!A'}$.

Conversely, suppose that $\Phi$ is causal and non-signalling. Then it satisfies the two non-signalling equations in Definition \ref{def:non-signalling}. Hence by \Ce, it can be factored in two ways:
\ctikzfig{two-causal-factorizations}
for causal processes $\Phi_i, \Psi_i$.

Now, take any effect $w: ((\bm A\multimap \bm A')\otimes (\bm B \multimap \bm B'))^* \cong (\bm B \multimap \bm B') \multimap (\bm A\multimap \bm A')^*$. For any causal state $\rho$,
\[ \input{phi2rho.tikz} \ \ :\ \bm B \limp \bm B' \]
Plugging this into one side of $w$ gives:
\[ \input{SOC2-state-splits0.tikz}\ \ :\ (\bm A\multimap \bm A')^* \]
Applying equation \eqref{eq:ho-splitting-explicit} gives:
\ctikzfig{SOC2-state-splits}
Hence by enough causal states we have
\ctikzfig{SOC2-splits}
It then follows that
\ctikzfig{SOC2Phi-1}
\ctikzfig{SOC2Phi-2}
Therefore $\Phi : ((\bm A\multimap \bm A')\otimes (\bm B \multimap \bm B'))^{**} = (\bm A\multimap \bm A')\otimes (\bm B \multimap \bm B')$.
\end{proof}

The proof above also generalises straightforwardly to show that a process is $n$-paritite non-signalling, i.e., for all $i$:
\ctikzfig{n-partite-NS}
if and only if:\vspace{-0pt}
\[ \Phi : (\bm A_1 \limp \bm A_1') \otimes
          (\bm A_2 \limp \bm A_2') \otimes
          \ldots \otimes
          (\bm A_n \limp \bm A_n')\vspace{-0pt} \]

\begin{theorem}
  For first-order systems $\bm A, \bm A', \bm B, \bm B'$, a process $\Phi$ is of type $(\bm A \limp \bm A') \parr (\bm B \limp \bm B')$ if and only if it is causal. That is:\vspace{-0pt}
  \[ (\bm A \limp \bm A') \parr (\bm B \limp \bm B') \cong
     \bm A \otimes \bm B \limp \bm A' \otimes \bm B' \vspace{-0pt}\]
\end{theorem}

\begin{proof}
  We rely on the relationship between $\limp$ and $\parr$:\vspace{-0pt}
  \begin{align*}
    (\bm A \limp \bm A') \parr (\bm B \limp \bm B')
    & \cong \bm A^* \parr \bm A' \parr \bm B^* \parr \bm B' \\
    & \cong \bm A^* \parr \bm B^* \parr \bm A' \parr \bm B' \\
    & \cong (\bm A^* \parr \bm B^*)^* \limp \bm A' \parr \bm B' \\
    & \cong \bm A \otimes \bm B \limp \bm A' \parr \bm B'\vspace{-0pt}
  \end{align*}
  Then, since $\bm A'$ and $\bm B'$ are first-order, $\bm A' \parr \bm B' \cong \bm A' \otimes \bm B'$, which completes the proof.
\end{proof}

So $(\bm A \limp \bm A') \parr (\bm B \limp \bm B')$ forms the joint system consisting of \textit{all} causal processes from $\bm A \otimes \bm B$ to $\bm A' \otimes \bm B'$, including the signalling ones, e.g.\vspace{-0pt}
\[
\input{swap.tikz}\ \ :\ (\bm A \limp \bm B) \parr (\bm B \limp \bm A)\vspace{-0pt}
\]

Hence, $\otimes$ and $\parr$ represent two extremes by which $\bm A \limp \bm A'$ and $\bm B \limp \bm B'$ can be combined, namely by requiring them to be non-signalling or imposing no non-signalling conditions. In the next section, we will see how to recover types that live in between these two extremes: the one-way signalling processes.

\subsection{Linear causal orders via combs}

\begin{theorem}\label{n=2 non signalling}
For first order systems $\bm A, \bm A', \bm B, \bm B'$,
a process $w$ is one-way signalling ($\evt{A} \preceq \evt{B}$) if and only if:\vspace{-0pt}
\[
\input{bi-w.tikz}\ \ :\ \bm A \multimap (\bm A' \multimap \bm B) \multimap \bm B'\vspace{-0pt}
\]
\end{theorem}
\begin{proof}
Suppose $\Phi$ is $\evt A \preceq \evt B$. First, we deform $\Phi$ to put the two $A$-labelled systems below the two $B$-labelled systems:
\[ \input{bi-w.tikz} \ \ \mapsto\ \ \input{bi-phi2.tikz} \vspace{-0pt}\]
The one-way signalling equation \eqref{eq:a-before-b} then becomes:\vspace{-0pt}
\begin{equation}\label{eq:bi-w-ns}
  \input{bi-phi-ns.tikz}\vspace{-0pt}
\end{equation}
Now, by \eqref{eq:lolli-sym} we have:\vspace{-0pt}
\[ \bm A \limp (\bm A' \limp \bm B) \limp \bm B' 
   \cong (\bm A' \limp \bm B) \limp \bm A \limp \bm B' \vspace{-0pt}\]
So, for $w$ to be a member of the above type, it must send all causal processes $\Psi : \bm A' \multimap \bm B$ to causal processes. This immediately follows from \eqref{eq:bi-w-ns}:\vspace{-0pt}
\ctikzfig{bi-phi-ns1}
Conversely, if $w$ sends causal processes to causal processes, it factorises as in \Cd, which implies \eqref{eq:bi-w-ns}:\vspace{-0pt}
\[\input{bi-phi-ns2.tikz}
  \tag*{\qedhere}
\]
\end{proof}

Hence, one-way signalling admits a more general class of processes than (two-way) non-signalling.

\begin{remark}\label{rem:type-embedding}
One can show the embedding of non-signalling processes into one-way signalling processes purely at the level of types by relying on the \textit{linear distributivity} property of $*$-autonomous categories \cite{Cockett1997133}. Namely, in any $*$-autonomous category, there exists a canonical mapping:\vspace{-0pt}
\[ (\bm X \parr \bm Y) \otimes \bm Z \to \bm X \parr (\bm Y \otimes \bm Z) \vspace{-0pt}\]
We can use this to construct an embedding of non-signalling processes into one-way signalling processes:\vspace{-0pt}
\begin{align*}
  (\bm A \limp \bm A') \otimes (\bm B \limp \bm B')
  & \cong (\bm A^* \parr \bm A') \otimes (\bm B^* \parr \bm B') \\
  & \to \bm A^* \parr (\bm A' \otimes (\bm B^* \parr \bm B')) \\
  & \to \bm A^* \parr (\bm A' \otimes \bm B^*) \parr \bm B' \\
  & \cong \bm A \limp (\bm A' \limp \bm B) \limp \bm B'\vspace{-0pt}
\end{align*}
Similarly, we can use the embedding $\bm X \otimes \bm Y \to \bm X \parr \bm Y$ noted in section \ref{sec:causC} to show that one-way signalling processes embed in the type $(\bm A \limp \bm A') \parr (\bm B \limp \bm B')$ of all processes.
\begin{align*}
  \bm A \limp (\bm A' \limp \bm B) \limp \bm B'
  & \cong \bm A^* \parr ((\bm A')^* \parr \bm B)^* \parr \bm B' \\
  & \cong \bm A^* \parr (\bm A' \otimes \bm B^*) \parr \bm B' \\
  & \to \bm A^* \parr \bm A' \parr \bm B^* \parr \bm B' \\
  & \cong (\bm A \limp \bm A') \parr (\bm B \limp \bm B') \\
\end{align*}
\end{remark}

Following \cite{PauloHierarchy}, we generalise from 2-party, one-way signalling processes to $n$-party processes by recursively defining a the type of $n$-\textit{combs}.

 

\begin{definition}
The \emph{$n$-combs} $C_n$ are defined by\vspace{-0pt}
\begin{itemize}
\item $C_0 = \bm I$,
\item $C_{i+1} = \bm B_{-i} \multimap C_i \multimap \bm B_{i+1}$\vspace{-0pt}.
\end{itemize}
\end{definition}

A 1-comb has type $\bm B_0 \limp \bm I \limp \bm B_1 \cong \bm B_0 \multimap \bm B_1$, so it is just a causal process. For higher combs, the `$-i$' is employed to maintain the left-to-right order of indices. For example, a 3-comb has type:\vspace{-0pt}
\[ \input{3-comb-bigger.tikz} : \bm B_{-2} \limp (\bm B_{-1} \limp (\bm B_0 \limp \bm B_1) \limp \bm B_2) \limp \bm B_3 \vspace{-0pt}\]
When necessary, we rename $\bm A_i := \bm B_{2i - n - 1}$ and $\bm A_i' := \bm B_{2i - n}$ to obtain e.g.,

\[ \bm A_1 \limp (\bm A_1' \limp (\bm A_2 \limp \bm A_2') \limp \bm A_3) \limp \bm A_3' \vspace{-0pt}\]

This recursive definition carries the following intuition. If we think of an $n$ comb as a communication protocol with $n$ input/output steps, then an $(n+1)$-comb is an $(n+1)$ step communication protocol, namely something which takes an initial input, then runs an $n$-step communication protocol, and produces a final output. Hence, we can represent the overall process of two agents running a communication protocol by plugging together Alice's $n$-comb and Bob's $(n-1)$-comb:\vspace{-0pt}
\[ \input{n-comb-plug2.tikz}\ :\ \ \bm A_1 \limp \bm A'_n \vspace{-0pt}\]
We give an alternative characterisation for combs in \CausC, which will relate to one-way signalling processes.

\begin{lemma}\label{w and w'}
  For any $n$-comb $w : C_n$, discarding the output $\bm A'_n$ separates as follows, for some $w'$:\vspace{-0pt}
  \begin{equation}\label{eq:n-comb-sep}
    \input{n-sep.tikz}
  \end{equation}
\end{lemma}
\begin{proof}
Plugging any causal state into the first input of $w$ and discarding the last output yields:\vspace{-0pt}
\[
\input{n-sep-lhs.tikz}\ \ :\ C_{n-1} \limp \bm I\vspace{-0pt}
\]
Then:\vspace{-0pt}
\begin{align*}
  C_{n-1} \limp \bm I  \cong C_{n-1}^* &  \cong (\bm B_{-(n-2)} \limp C_{n-2} \limp \bm B_{n-1})^* \\
  & \cong (\bm B_{-(n-2)} \otimes C_{n-2} \limp \bm B_{n-1})^*\vspace{-0pt}
\end{align*}
Hence by Lemma~\ref{lem:ho-splitting}, in particular equation \eqref{eq:ho-splitting-explicit}, we obtain:
\ctikzfig{n-sep-lhs-sep}
The result then follows from enough causal states.
\end{proof}

Note that we haven't actually said that $w'$ is itself an $(n-1)$-comb. We will show this now.

\begin{theorem}\label{n comb to n-1 comb}
$w$ is an $n$-comb, i.e., $w : C_n$, if and only if it separates as in equation \eqref{eq:n-comb-sep} for some $w' : C_{n-1}$.
\end{theorem}
\begin{proof}
By induction. For $n=1$ the theorem is true because a $0$-comb is always $\bm I$.
Suppose the theorem is true for $n$.
Let $w$ be an $(n+1)$-comb. We need to show that $w'$ is an $n$-comb.
So let $y$ be any $(n-1)$-comb. Then, if we form the process:\vspace{-0pt}
\begin{equation}\label{eq:ih-y}
  \input{IndHyp-y.tikz}\vspace{-0pt}
\end{equation}
then clearly discarding the top output results in an $(n-1)$ comb (namely $y$) and a discard.
So by the induction hypotheses, \eqref{eq:ih-y} is an $n$-comb. Therefore we have\vspace{-0pt}
\ctikzfig{IndHyp-y2}\vspace{-0pt}
where $(*)$ follows from the definition of $(n+1)$-comb and $(**)$ is Lemma \ref{w and w'}. Hence $w'$ sends any $(n-1)$-comb to a causal map, so $w'$ is itself an $n$-comb.

Conversely, let $w'$ in equation \eqref{eq:n-comb-sep} be an $n$-comb, and take any $n$-comb $y$.
Then by the induction hypothesis, discarding the top output of $y$ separates as discarding and an $(n-1)$-comb $y'$. Hence:\vspace{-0pt}
\ctikzfig{n-comb-trace}\vspace{-0pt}
so $w$ is an $(n+1)$-comb.
\end{proof}

Hence, $n$-combs can be characterised inductively in exactly the same way as $n$-party one-way signalling processes. Since $1$-combs are just causal processes, the following is immediate.

\begin{corollary}
For first order systems $\bm A_1, \bm A'_1, \ldots, \bm A_n, \bm A'_n$, a map $w:A_1 \otimes \ldots \otimes A_n \to A'_1 \otimes \ldots \otimes A'_n$ is one-way signalling ($\evt{A}_1 \preceq \ldots \preceq \evt{A}_n$) if and only if it is of type $\bm A_1 \multimap ( \bm A'_1 \multimap ( \ldots ) \multimap \bm A_n ) \multimap \bm A'_n$.
That is, it is an $n$-comb.
\end{corollary}

Stated in terms of consistency with causal orderings, as defined in section~\ref{sec:disc}, the above corollary says that $\Phi : \bm A_1 \multimap ( \bm A'_1 \multimap ( \ldots ) \multimap \bm A_n ) \multimap \bm A'_n$ if and only if:
\[
\input{nfold.tikz}\ \vDash\ \input{lin-ord.tikz}
\]



Proposition~\ref{w factorization}, then generalises the characterisation theorem for quantum combs in \cite{PaviaIndef} to \CausC for any precausal $\mathcal C$: an $n$-comb always factors as a sequence of `memory channels', i.e., a composition of causal processes of the form:
\ctikzfig{ch-w-mem}

\subsection{Directed acyclic causal orders via pullback}\label{sec:pullback}\ \\

\noindent 
We now extend from types which express consistency with arbitrary, linear causal orderings to all acyclic causal orderings. To do so, we first give a new characterisation of the consistency relation $\Phi \vDash \mathcal O$ defined in section~\ref{sec:disc} in terms of all of the total orders refining the partial order $\mathcal O$. Recall the following standard definition:

\begin{definition}
For a partial order $\mathcal O$, a \textit{totalisation} $\mathcal O'$ is a total order with the same elements as $\mathcal O$ such that $e \preceq_{\mathcal O} e' \implies e \preceq_{\mathcal O'} e'$.

\end{definition}

Totalisations are not unique. For example, a typical `common cause' ordering admits two possible totalisations:
\[
\input{common-cause.tikz}
\qquad
\textrm{\large $\leadsto$}
\qquad
\input{common-cause-lin1.tikz}
\qquad
\textrm{and}
\qquad
\input{common-cause-lin2.tikz}
\]
corresponding to imposing the (extraneous) causal orderings $\evt{B} \preceq \evt{C}$ or $\evt{C} \preceq \evt{B}$.
Its important to note that causal constraints come from the \textit{absence} of edges in the Hasse diagrams (i.e., diagrams to visualize a partial order) above, rather than the presence. That is, the absence of an edge between $\evt{B}$ and $\evt{C}$ in the leftmost diagram asserts the independence of the output of $\evt{B}$ from the input of $\evt{C}$ and vice-versa. Hence, consistency with $\mathcal O$ implies consistency with any totalisation of $\mathcal O$.

\begin{lemma}\label{lem:o-implies-total}
  For a process $\Phi$, a causal ordering $\mathcal O$, and a totalisation $\mathcal O'$ of $\mathcal O$, we have $\Phi \vDash \mathcal O \implies \Phi \vDash \mathcal O'$.
\end{lemma}

\begin{proof}
  Assume $\Phi \vDash \mathcal O$. Then, for any subset $\mathcal E \subseteq \mathcal O$, there exists a process $\Phi'$ such that:
  \begin{equation}\label{eq:sat-O}
    \input{satisfies-O.tikz}
  \end{equation}
  But since $e \preceq_{\mathcal O} e' \implies e \preceq_{\mathcal O'} e'$, we have $\past_{\mathcal O}(\mathcal E) \subseteq \past_{\mathcal O'}(\mathcal E)$. Hence \eqref{eq:sat-O} implies that $\Phi$ factors as required by Definition~\ref{def:causal-consistency}:
  \begin{equation*}
    \input{satisfies-Op.tikz}
  \end{equation*}
  Therefore $\Phi \vDash \mathcal O'$.
\end{proof}

Combining this with its converse yields the following theorem.

\begin{theorem}\label{thm:totalisation}
  For a process $\Phi$ and a causal ordering $\mathcal O$, $\Phi \vDash \mathcal O$ if and only if, for every totalisation $\mathcal O'$ of $\mathcal O$, $\Phi \vDash \mathcal O'$.
\end{theorem}

\begin{proof}
 $(\Rightarrow)$ follows immediately from Lemma~\ref{lem:o-implies-total}. 
For $(\Leftarrow)$, let $\mathcal E \subseteq \mathcal O$ be any subset.
Split $\mathcal O$ into two parts, $\mathcal O_1 := \past_{\mathcal O}(\mathcal E)$ 
and $\mathcal O_2 := \mathcal O \backslash \past_{\mathcal O}(\mathcal E)$.
Define a total ordering $\mathcal O'$ on $\mathcal O_1 \cup \mathcal O_2$ 
by requiring that every element in $\mathcal O_1$ is below every element in $\mathcal O_2$ 
and taking any totalization on $\mathcal O_1$ and $\mathcal O_2$.
This order refines $\mathcal O$ because $\past_{\mathcal O}(\mathcal E)$ is downward-closed, 
and by construction $\past_{\mathcal O}(\mathcal E) = \past_{\mathcal O'}(\mathcal E)$. 
Hence $\Phi \vDash \mathcal O'$ implies:

  \ctikzfig{satisfies-OOp}
  Since we can find a suitable totalisation to give the equation above for any subset $\mathcal E \subseteq \mathcal O$, we have $\Phi \vDash \mathcal O$.
\end{proof}

Hence, we can always replace consistency with a causal ordering with consistency with all of its totalisations, i.e., all of the linear orders which refine it. For example:
\begin{equation}\label{eq:totalise-order}
\input{big-phi-3.tikz} \ \vDash\ \input{common-cause.tikz}
\quad\iff \quad
\left(\ \ 
\input{big-phi-3.tikz} \ \vDash\ \input{common-cause-lin1.tikz}
\quad \wedge \quad 
\input{big-phi-3.tikz} \ \vDash\ \input{common-cause-lin2.tikz}
\ \ \right)
\end{equation}
We already showed in the previous section that we can capture any linear ordering via combs. Hence, to capture any causal ordering, we only need to be able to take conjunctions of linear causal orders. This can be accomplished via \textit{intersections} of types, much like those defined in \cite{PauloHierarchy} for the quantum case.

The types on the RHS of \eqref{eq:totalise-order} are, respectively:
\[
\bm X := \bm A \limp (\bm A' \limp (\bm B \limp \bm B') \limp \bm C) \limp \bm C'
\qquad\textrm{and}\qquad
\bm X' := \bm A \limp (\bm A' \limp (\bm C \limp \bm C') \limp \bm B) \limp \bm B'
\]
Our goal now is to form the intersection $\bm X \cap \bm X'$, which will consist precisely of those causal processes consistent with the causal ordering given by $\bm X$ and the one given by $\bm X'$. Categorically, intersections are given by pullback, so our first goal will be to find an object in which both $\bm X$ and $\bm X'$ embed. Since combs are special cases of causal processes, we can embed both $\bm X$ and $\bm X'$ into the system of \textit{all} causal processes of the appropriate type, namely:
\[ \bm X \hookrightarrow (\bm A \otimes \bm B \otimes \bm C \limp \bm A' \otimes \bm B' \otimes \bm C') \hookleftarrow \bm X' \]
In fact, we can always find such embeddings for any system built inductively from first-order systems:

\begin{theorem}\label{thm:fo-embed}
  Any object $\bm X$ built inductively from first order systems has a canonical embedding of the form:
  \[ e : \bm X \rightarrow (\bm A_1 \otimes \ldots \otimes \bm A_m \limp \bm A_1' \otimes \ldots \otimes \bm A_n') \]
  for some first-order systems $\bm A_1, \ldots, \bm A_m, \bm A_1', \ldots, \bm A_n'$. Furthermore, the underlying $\mathcal C$-morphism is just a permutation of systems.
\end{theorem}

The proof is a straightforward application of the embedding $\bm A \otimes \bm B \hookrightarrow \bm A \parr \bm B$ and the special property of first order systems that $\bm A \otimes \bm B \cong \bm A \parr \bm B$. It is given explicitly in the appendix.

We will construct $\bm X \cap \bm X'$ essentially in terms of a set-theoretic intersection of their associated states $c_{\bm X}$ and $c_{\bm X'}$. For that, we we make use of the following lemma.

\begin{lemma}\label{lem:intersect}
Let $c$ and $d$ be sets of states for the same object $A$ which are flat, closed, and furthermore satisfy the property that $\lambda \mmix \in c$ and $\lambda \mmix \in d$ for a fixed scalar $\lambda$. Then $c \cap d$ is also flat and closed.
\end{lemma}
\begin{proof}
For both properties, we rely on the fact that $(-)^*$ is order-reversing. That is, $a \subseteq b \Rightarrow b^* \subseteq a^*$. For flatness, we have by assumption that $\lambda \mmix \in c \cap d$. Since $c \cap d \subseteq c$, we have $c^* \subseteq (c \cap d)^*$. So by flatness of $c$, $\mu \disc \in c^*$ for some $\mu$. Hence $\mu \disc \in (c \cap d)^*$ and $c \cap d$ is flat.

For closure, first note that any set of states is contained in its double dual, so we have $c \cap d \subseteq (c \cap d)^{**}$.
For the converse, $c \cap d \subseteq c$ implies $c^* \subseteq (c \cap d)^*$ and similarly $d^* \subseteq (c \cap d)^*$. Hence, $c^* \cup d^* \subseteq (c \cap d)^*$, so $(c \cap d)^{**} \subseteq (c^* \cup d^*)^*$. It therefore suffices to show that $(c^* \cup d^*)^* \subseteq c \cap d$. This follows from the fact that $c^* \subseteq c^* \cup d^*$, so $(c^* \cup d^*)^* \subseteq c^{**} = c$ and similarly $(c^* \cup d^*)^* \subseteq d^{**} = d$.
\end{proof}



\begin{theorem}\label{thm:pullback}
Let $\bm X, \bm X'$ be a pair of objects with canonical embeddings $e, f$ into a fixed system $\bm Y := \bm A_1 \otimes \ldots \otimes \bm A_m \limp \bm A_1' \otimes \ldots \otimes \bm A_n'$,
as in Theorem~\ref{thm:fo-embed}. Then there exists an object $\bm X \cap \bm X'$ and morphisms $p_1, p_2$ in $\Caus[\mathcal C]$ making the following pullback:
\begin{equation}\label{eq:pullback}
\vcenter{\xymatrix{
\bm X \cap \bm X' \ar[d]^{p_2} \ar[r]^{p_1} & \bm X \ar[d]^{e} \\
\bm X' \ar[r]^{f}  & \bm Y
}}
\end{equation}
\end{theorem}
\begin{proof}
Let $\bm Y := (Y, c_{\bm Y})$ and define the following two sets of states for $\bm Y$:
\begin{align*}
\overline{c}_{\bm X} & := \{ e \circ \rho \, | \, \rho \in c_{\bm X} \} \\
\overline{c}_{\bm X'} & := \{ f \circ \rho \, | \, \rho \in c_{\bm X'} \}
\end{align*}
Since $e$ and $f$ are just permutations of systems, it is straightforward to show that both of these sets are flat, closed, and  both contain $\lambda \mmix$ for some fixed $\lambda$. Hence, applying Lemma~\ref{lem:intersect}, we have that $\overline{c}_{\bm X} \cap \overline{c}_{\bm X'}$ is flat and closed. Then, let $\bm X \cap \bm X' := (Y, \overline{c}_{\bm X} \cap \overline{c}_{\bm X'})$, $p_1 := e^{-1}$ and $p_2 := f^{-1}$. It is straightforward to check that $p_1, p_2$ are indeed $\Caus[\mathcal C]$-morphisms and diagram \eqref{eq:pullback} clearly commutes.

It only remains to show that, for any $g: \bm Z \to \bm X$ and $h: \bm Z \to \bm X'$ such that $e \circ g = f \circ h$, there is a unique mediating morphism $z : \bm Z \to \bm X \cap \bm X'$:
\begin{equation*}
\xymatrix{
\bm Z \ar[rrd]^{g} \ar[rdd]_{h} \ar@{.>}[rd]|-{\,z\,} & & \\
& \bm X \cap \bm X' \ar[d]^{f^{-1}} \ar[r]_{e^{-1}} & \bm X \ar[d]^{e} \\
& \bm X' \ar[r]^{f}  & \bm Y
}
\end{equation*}
Since $e$ and $f$ are isomorphisms, the only possibility is $z := e \circ g = f \circ h$. So, it suffices to show that $z$ is a morphism in $\Caus[\mathcal C]$. For any $\rho \in c_{\bm Z}$, $g \circ \rho \in c_{\bm X}$, so $z \circ \rho = e \circ g \circ \rho \in \overline{c}_{\bm X}$. Similarly, $z \circ \rho = f \circ h \circ \rho \in \overline{c}_{\bm X'}$. Hence $z \circ \rho \in \overline{c}_{\bm X} \cap \overline{c}_{\bm X'}$, which completes the proof.
\end{proof}


Since combs embed into all causal processes, it immediately follows that we can take intersections of combs via pullback. Hence, the condition on $\Phi$ stated in \eqref{eq:totalise-order} can equivalently be given in terms of $\Phi$ having the following type:
\[
\input{big-phi-3.tikz} : (\bm A \limp (\bm A' \limp (\bm B \limp \bm B') \limp \bm C) \limp \bm C') \cap (\bm A \limp (\bm A' \limp (\bm C \limp \bm C') \limp \bm B) \limp \bm B')
\]
This generalises in the obvious way to any causal ordering $\mathcal O$.


The fact that $\cap$ arises as a pullback also gives us some properties of intersections `for free'. For instance, any functor with a left adjoint necessarily preserves limits. By the definition of $*$-autonomous categories, $(\bm B \limp -)$ has a left adjoint given by $(- \otimes \bm B)$, so the following is immediate:

\begin{corollary}
For objects $\bm A, \bm A'$ and $\bm B$ in \CausC we have
\[
\bm B \multimap (\bm A \cap \bm A') \cong (\bm B \limp \bm A) \cap (\bm B \limp \bm A')
\]
\end{corollary}


\section{Higher-order systems with indefinite causal order}\label{sec:indef}


Quantum theory, as it is typically understood, assumes a fixed background
time, and hence a fixed causal ordering. However, in the context of higher-
order processes, this restriction is not necessary to obtain a theory which
behaves \textit{locally} like classical or quantum theory. In this section we
shall take a look at \emph{process matrices}, introduced
in~\cite{ViennaIndef}, to investigate processes which do not have a definite
causal order. Such processes were called bipartite second-order causal
in~\cite{EPTCS236.6}.

\begin{definition}
A process $w : (A^{*} \otimes A') \otimes (B^{*} \otimes B') \to C^{*} \otimes C$
is called \emph{bipartite second-order causal} ($\textrm{SOC}_2$)
if for all causal $\Phi_{A}, \Phi_{B}$ the following map is causal:\vspace{-0pt}
  \ctikzfig{SOC2}\vspace{-0pt}
\end{definition}

So $\textrm{SOC}_2$ maps send products of causal processes to a causal process.
The following shows that $\textrm{SOC}_2$ processes are actually normalized on \emph{all}
non-signalling maps, not just product maps.

\begin{theorem}\label{thm:soc2-char}
For first order systems $\bm A, \bm A', \bm B, \bm B', \bm C, \bm C'$,
a process $w$ is $\textrm{SOC}_2$ if and only if it is of type 
$(\bm A \multimap \bm A') \otimes (\bm B \multimap \bm B') \multimap (\bm C \multimap \bm C')$.
\end{theorem}

\begin{proof}
Since products of causal processes are non-signalling, they are in $(\bm A \multimap \bm A') \otimes (\bm B \multimap \bm B')$, so any process of the above type is indeed $\textrm{SOC}_2$.

For the converse, let $\pi$ be an effect of type $(\bm C \multimap \bm C')^*$.
Then $\pi \circ w$ is an effect on products of causal processes.
Now Lemma~\ref{lem:tensor-effects} states that $\pi \circ w$ yields $1$ for product states if and only if it yields $1$ for any state in the tensor product. Hence it is an effect for $(\bm A \multimap \bm A') \otimes (\bm B \multimap \bm B')$,
which by Theorem~\ref{type of non-signalling} are precisely the non-signalling maps.
By Proposition~\ref{prop:morphism} this means $w: (\bm A \multimap \bm A') \otimes (\bm B \multimap \bm B') \multimap (\bm C \multimap \bm C')$
\end{proof}

This represents a significant strengthening of the result in~\cite{EPTCS236.6}, which was only able to show that $\textrm{SOC}_2$ extends to all so-called \textit{strongly non-signalling} processes, which are a special case of non-signalling processes.

Special cases of $\textrm{SOC}_2$ processes are 3-combs which arise from fixing a causal ordering between $\bm A$ and $\bm B$:\vspace{-0pt}
\begin{eqnarray*}
   \bm C \limp (\bm A \limp (\bm A' \limp \bm B) \limp \bm B') \limp \bm C'\\
   \bm C \limp (\bm B \limp (\bm B' \limp \bm A) \limp \bm A') \limp \bm C'\vspace{-0pt}
\end{eqnarray*}
Indeed one can show the containment of either of these types into the type of $\textrm{SOC}_2$ processes using a simple calculation on types much like in Remark~\ref{rem:type-embedding}. However, the most interesting $\textrm{SOC}_2$ processes are those which do not arise from combs.

\begin{example}
The \textit{OCB process} is defined as follows:\vspace{-0pt}
\ctikzfig{Wdef}\vspace{-0pt}
where $\sigma_x, \sigma_z$ are Pauli matrices and associated effects. Note that, while the individual summands are not positive, the result is, yielding a process in \CPM. The fact that it is an $\textrm{SOC}_2$ process in $\Caus[\CPM]$ follows straightforwardly from the fact that the Pauli matrices are trace-free. Furthermore, it was shown in~\cite{ViennaIndef} that it can be used to win certain non-local games with higher probability than any causally-ordered process, due to the fact that Bob can, to some extent, choose a causal ordering between himself and Alice \textit{a posteriori} by his choice of quantum measurement.
\end{example}

Theorem \ref{thm:soc2-char} extends naturally to a characterisation of $n$-partite second-order causal processes ($\textrm{SOC}_n$) via:\vspace{-0pt}
\[ (\bm A_1 \limp \bm A_1') \otimes \ldots \otimes (\bm A_n \limp \bm A_n') \limp (\bm C \limp \bm C') \vspace{-0pt}\]

\begin{example}
  It is not necessary to go to $\Caus[\CPM]$ to find processes exhibiting indefinite causal structure. Indeed the following process:
  \ctikzfig{BandW}
  is an $\textrm{SOC}_3$ process in $\Caus[\MatRp]$, where the `$-$' labelled states and effects
  are column vectors and row vectors with values $(1,-1)$ , respectively.
 It was shown in \cite{WolfClassicalMulti} that this process, as well as a generalisation to an $\textrm{SOC}_n$ process for odd $n$, was incompatible with any pre-defined causal order.
\end{example}

An interesting family of $\textrm{SOC}_2$ processes are \textit{switches}, where an auxiliary system is used to control the causal ordering of the input processes.

\begin{definition}
  For first-order systems $\bm X$ and $\bm A = \bm A' = \bm B = \bm B' = \bm C = \bm C'$, a \textit{switch} is a process of type:\vspace{-0pt}
  \begin{equation}\label{eq:switch}
    \input{switch.tikz}\ \ : \bm X \otimes \bm C \limp (\bm A \limp \bm A') \otimes (\bm B \limp \bm B') \limp \bm C'\vspace{-0pt}
  \end{equation}
  in \CausC, such that for distinct states $\rho_0, \rho_1 : \bm X$, we have:\vspace{-0pt}
  \begin{equation}\label{eq:switch-eqs}
    \input{switch0.tikz} \qquad \input{switch1.tikz}\vspace{-0pt}
  \end{equation}
\end{definition}


We now see some concrete examples of the switch, as a higher-order stochastic map and as a quantum channel.

\begin{example}\label{ex:class-switch}
  For $\mathcal C = \MatRp$, the \textit{classical switch} process is uniquely fixed by \eqref{eq:switch-eqs} if we let $\bm X = 2$ and:\vspace{-0pt}
  \[
  \rho_0 := \left(\begin{matrix}
    1 \\ 0
  \end{matrix}\right)
  \qquad
  \rho_1 := \left(\begin{matrix}
    0 \\ 1
  \end{matrix}\right) \vspace{-0pt}\]
  Indeed, $s$ is given by:\vspace{-0pt}
  \begin{equation}\label{eq:switch-class}
    \input{switch-class.tikz}\vspace{-0pt}
  \end{equation}
  where $\rho_i' := \rho_i^T$. Then, since $\rho_0' + \rho_1' = \disc$, we have:
  \ctikzfig{switch-class-caus}
  \ctikzfig{switch-class-caus2}\vspace{-0pt}
  Hence $s$ has the correct type shown in \eqref{eq:switch}.
\end{example}

\begin{example}
  For $\mathcal C = \CPM$, a switch process can be defined just as in \eqref{eq:switch-class}, by letting $\bm X = \mathcal L(\mathbb C^2)$ and replacing $\rho_i$ and $\rho_i^T$ with the appropriate qubit projections and their associated quantum effects:\vspace{-0pt}
  \[ \rho_i := \ketbra{i}{i} \qquad\qquad \rho_i'(\mu) := \textrm{Tr}(\ketbra{i}{i}\mu) \vspace{-0pt}\]
  This is precisely the $\mathcal Z$ superoperator defined in \cite{QSwitch}, which defines a (decoherent) switch for quantum channels.

  However, unlike Example \ref{ex:class-switch}, this channel is \textit{not} uniquely fixed by \eqref{eq:switch-eqs}, since $\rho_0, \rho_1$ do not form a basis for $\mathcal L(\mathbb C^2)$. 
  For instance, plugging $\rho := \ketbra{+}{+}$ into $\bm X$ of this process yields a classical mixture of the two possible wirings:
  \ctikzfig{switch-mix}
One can also define a \textit{coherent} quantum switch which does satisfy \eqref{eq:switch-eqs}, but where inputting the state $\ketbra{+}{+}$ into $\bm X$ yields a quantum superposition of causal orderings. See \cite{QSwitch} for details.
\end{example}

\begin{theorem}
  A switch cannot be causally ordered. That is, the type of $s$ does not restrict to one of the following:\vspace{-0pt}
  \begin{eqnarray*}
   \bm X \otimes \bm C \limp (\bm A \limp (\bm A' \limp \bm B) \limp \bm B') \limp \bm C'\\
  \bm X \otimes \bm C \limp (\bm B \limp (\bm B' \limp \bm A) \limp \bm A') \limp \bm C'\vspace{-0pt}
  \end{eqnarray*}
  unless $\bm A \cong \bm I$.
\end{theorem}

\begin{proof}
  Suppose $s$ is causally ordered with $\evt A \preceq \evt B$. That is:
  \[ s : \bm X \otimes \bm C \limp (\bm A \limp (\bm A' \limp \bm B) \limp \bm B') \limp \bm C'\vspace{-0pt} \]
  Plugging states and effects into $s$ yields a simpler type:
  \begin{equation}\label{eq:switch-plug}
    \input{switch-plug.tikz}\ \ :\ (\bm A \limp (\bm A' \limp \bm B) \limp \bm B')^*
  \end{equation}
  By Theorem \ref{n=2 non signalling} characterising one-way signalling processes, one can verify that for any $\Phi : \bm A \limp \bm B'$, we have:\vspace{-0pt}
  \[ \input{signal-single.tikz}\ \ :\ \bm A \limp (\bm A' \limp \bm B) \limp \bm B' \vspace{-0pt}\]
  Since this is the dual type of \eqref{eq:switch-plug}, composing the two yields\vspace{-0pt}
  \[ \input{switch-plug1.tikz}\ =\ \input{time-travel2.tikz} \vspace{-0pt}\]
  which violates no time-travel, Theorem \ref{thm:no-time-travel}. Hence $\bm A \cong \bm I$. The second causal ordering can be ruled out symmetrically, by plugging the state $\rho_0$ into $s$.
\end{proof}


\section{Conclusion and Future Work}


In order to study higher order processes, we have created a categorical construction
which sends certain compact closed categories $\mathcal C$ to a new category $\CausC$.
There is a fully faithful embedding of the category of first order causal processes of $\mathcal C$ into $\CausC$,
but we are also able to talk about genuine higher order causal processes.
This new category also has a richer structure which allows us to develop a type theory for its objects.
We classify certain kinds of processes in this type theory,
such as the non-signalling processes, one-way signalling processes, combs and bipartite second order causal processes and show that the type theoretic characterisation of these processes coincides with the operational one involving discarding.

The construction of \CausC can be generalised straightforwardly to encompass `sub-causal' processes as well, by replacing the definition of $(-)^*$ with:\vspace{-0pt}
\[ c^* = \{\pi:A^* \mid \forall x\in c, \pi \circ x \in \mathcal{M}\} \vspace{-0pt}\]
for a suitable sub-monoid $\mathcal M$ of $\mathcal C(I,I)$ to get $\Caus_{\mathcal M}[\mathcal C]$.
Then, we recover \CausC as $\Caus_{\{1\}}[\mathcal C]$. However, we obtain other interesting examples by varying $\mathcal M$. In the case of $\MatRp$, $\mathcal I = \mathbb R_+$. Taking $\mathcal M$ to be the unit interval, $\Caus_{[0,1]}[\MatRp]$ gives us the full subcategory of probabilistic coherence spaces on finite sets~\cite{DanosCoherence}. Similarly, $\Caus_{[0,1]}[\CPM]$ has trace non-increasing CP maps as its first-order processes, and generalisations thereof at higher orders. Alternatively, we can build a category of `causal processes with failure' by letting $\mathcal M$ be $\{0,1\}$. Exploring the properties of these categories, and how they relate to \CausC is a subject of future work.
Another subject for future research is the relation between the types of causal systems and multiplicative linear logic (MLL).
Since $*$-autonomous categories are a model of MLL, MLL provides a (decidable) fragment of the logic of type containment in \CausC. This opens up possibilities to automate many proofs using existing automated linear logic provers. Indeed many of the type relationships in this paper were discovered with the help of such a tool, called \texttt{llprover} \cite{llprover} (see Fig.~\ref{fig:llprover}).

\begin{figure}
  \centering
  \includegraphics[width=\textwidth]{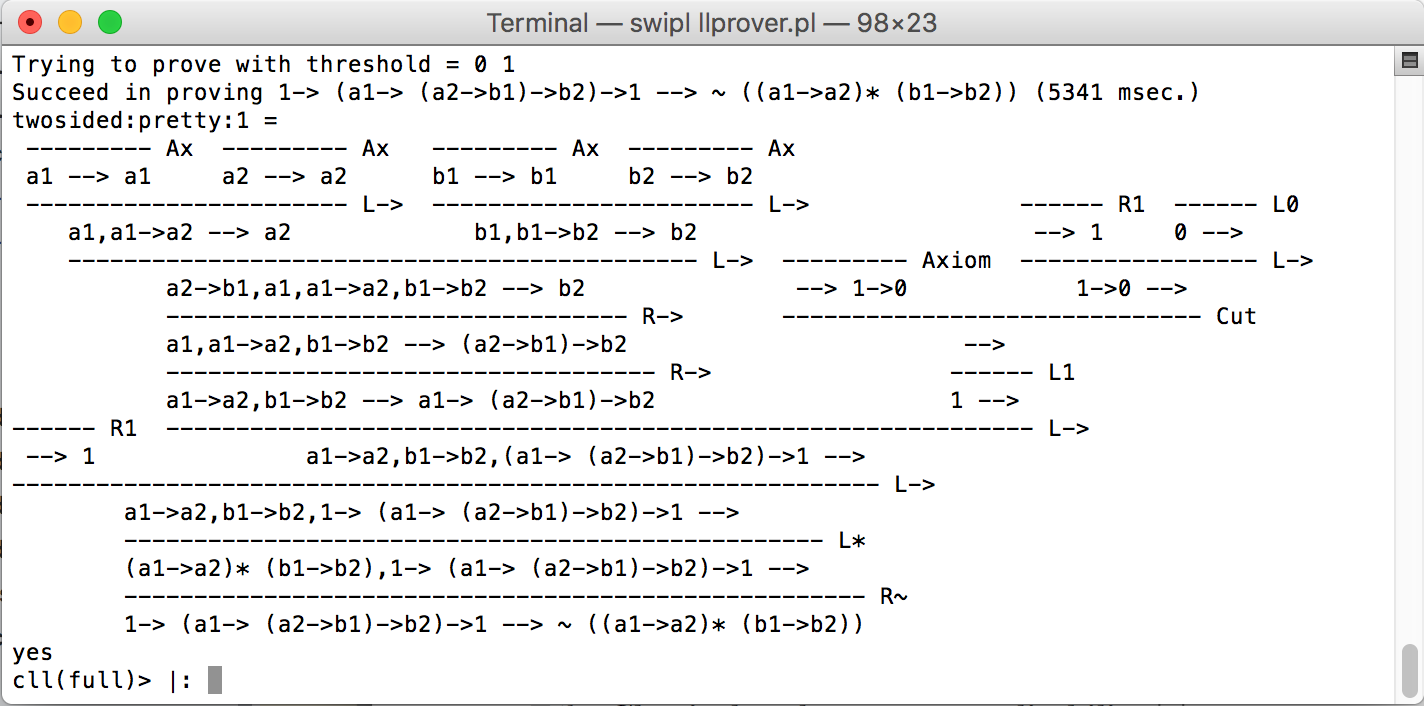}

  \caption{\label{fig:llprover} Proving a 3-comb with trivial input and output systems embeds in $\SOCt$ using \texttt{llprover}.}
\end{figure}

A third direction for future work is to understand exactly how expressive the internal logic of a HOCC is. 
We showed in this paper that it is possible to express generalised conditional independences of a `non-signalling' variety, 
which always takes a certain form for a process with inputs and outputs (e.g., a conditional probability distribution):
\[ \textrm{(outputs)} \indep \textrm{(inputs)} \ |\ \textrm{(other inputs)} \]
For example, the usual 2-party non-signalling conditions for a process $\Phi : A \otimes B \to A' \otimes B'$ are of the form:
\[
A' \indep B \ |\ A
\qquad\textrm{and}\qquad
B' \indep A \ |\ B 
\]
However, classical conditional indepedences can be of the form $\mathcal A \indep \mathcal B \ |\  \mathcal C$ for arbitrary sets of random variables $\mathcal A, \mathcal B, \mathcal C$. It is therefore natural to ask if we could express such indepedences within a HOCC, possibly with some extra structure. One possible approach is to inter-convert inputs and outputs via the probabilistic operation of \textit{disintegration}. This classical technique of converting a joint state into a reduced state and a conditional distribution is a crucial `subroutine' in Bayesian inference, and has recently given a categorical/string-diagrammatic characterisation by Cho and Jacobs~\cite{ChoJacobsDisint} and showed that categories that admit disintegration enable various equivalent expressions of (arbitrary) conditional independences of systems, and can even recover the graphoid axioms of Verma and Pearl~\cite{VermaPearl}. Thus, considering precausal categories with disintegration should enable one to say much more about conditional independences (and processes which preserve them) than just plain precausal categories. Unfortunately, disintegration in its usual form relies crucially on the copiability of classical data, so it does not have a straightforward quantum analogue. Nonetheless, quantum analogues to Bayesian conditioning~\cite{LeiferSpekkens} and more recently quantum common causes~\cite{AllenCommonCause} could provide a solution analogous to the classical case.

Pushing this a bit further, one not only wishes to talk about independences within a HOCC, but also draw causal conclusions. That is, one not only wishes to rule out possible causal explanations, but also one wants to say when a causal influence is present (and possibly even measure it). Exploring the connection with techniques in classical causal inference of a logical or equational flavour, such as Pearl's \textit{do-calculus}~\cite{PearlBook}, may give some clues.




\bibliographystyle{plain}
\bibliography{main}

\clearpage

\appendix

\section{Omitted proofs}

\subsection{Proof that \CausC is $*$-autonomous}\ \\

\noindent In this part of the appendix we will give full proofs leading up to the fact that \CausC is indeed a $*$-autonomous category.

\begin{theorem*}
  $\CausC$ is an SMC, with tensor given by:
  \[ \bm A \otimes \bm B := (A \otimes B, c_{\bm A \otimes \bm B}) \]
  and tensor unit $\bm I := (I, \{1\})$.
\end{theorem*}
\begin{proof}
  First we show that $\bm A \otimes \bm B$ and $\bm I$ are indeed objects in \CausC, namely the $c_{\bm A \otimes \bm B}$ and $c_{\bm I}$ are flat and closed. This is immediate for $c_{\bm I}$, so we focus on $c_{\bm A \otimes \bm B}$.

  Closure follows immediately from Lemma \ref{lem:triple-dual}, so it remains to show flatness. Since $c_{\bm A}$ and $c_{\bm B}$ are flat, then for some $\lambda, \lambda'$:
  \[ \lambda\ \maxmix \in c_{\bm A},\  \lambda'\ \maxmix \in c_{\bm B} \]
  hence:
  \[ \lambda\lambda'\ \maxmix\,\maxmix \in c_{\bm A} \otimes c_{\bm B} \subseteq c_{\bm A \otimes \bm B} \]
  Similarly, for some $\mu, \mu'$:
  \[ \mu\ \discard \in c_{\bm A}^*,\ \mu'\ \discard \in c_{\bm B}^* \]
  So, for all $\rho \in c_{\bm A}$, $\rho' \in c_{\bm B}$, we have:
  \[ \mu\mu'\ \weight{\rho}\, \weight{\rho'} = 1 \]
  which implies, by Lemma \ref{lem:triple-dual}:
  \[ \mu\mu' \discard \, \discard \in (c_{\bm A} \otimes c_{\bm B})^* =
   (c_{\bm A} \otimes c_{\bm B})^{***} =: c_{\bm A \otimes \bm B}^*\]
  Next, we show associativity and unit laws for $\otimes$. For any object $\bm A$, the unit laws $\bm A \otimes \bm I = \bm A = \bm I \otimes \bm A$ follow from the closure of $c_{\bm A}$. Associativity is a bit trickier. We first work in terms of effects in order to take advantage of Lemma \ref{lem:tensor-effects}. Applying this to an effect $\pi \in c_{(\bm A \otimes \bm B) \otimes \bm C}^*$ gives:
  \begin{align*}
  \quant{\forall\ \Psi \in c_{(\bm A \otimes \bm B) \otimes \bm C}}{\input{tri-braket1.tikz}\ =\ 1}
  & \ \ \iff\ \ 
  \quant{\forall\ \Psi \in c_{\bm A \otimes \bm B}, \xi \in c_{\bm C}}{\input{tri-braket2.tikz}\ =\ 1} \\[3mm]
  & \ \ \iff\ \ 
  \quant{\forall\ \psi \in c_{\bm A}, \phi \in c_{\bm B}, \xi \in c_{\bm C}}{\input{tri-braket3.tikz}\ =\ 1}
  \end{align*}
  Similarly, for $\pi \in c_{\bm A \otimes (\bm B \otimes \bm C)}^*$
  \begin{align*}
  \quant{\forall\ \Psi \in c_{\bm A \otimes (\bm B \otimes \bm C)}}{\input{tri-braket1.tikz}\ =\ 1}
  & \ \ \iff\ \ 
  \quant{\forall\ \psi \in c_{\bm A}, \Phi \in c_{\bm B \otimes \bm C}}{\input{tri-braket2p.tikz}\ =\ 1}  \\[3mm]
  & \ \ \iff\ \ 
  \quant{\forall\ \psi \in c_{\bm A}, \phi \in c_{\bm B}, \xi \in c_{\bm C}}{\input{tri-braket3.tikz}\ =\ 1} 
  \end{align*}
  Hence $c_{(\bm A \otimes \bm B) \otimes \bm C}^* = c_{\bm A \otimes (\bm B \otimes \bm C)}^*$ so $c_{(\bm A \otimes \bm B) \otimes \bm C} = c_{\bm A \otimes (\bm B \otimes \bm C)}$, which implies associativity of $\otimes$.

  Next we show that $\otimes$ is well-defined on morphisms. For morphisms $f : \bm A \to \bm A'$, $g : \bm B \to \bm B'$, and an effect $\pi \in c_{\bm A' \otimes \bm B'}^*$, we have by Lemma \ref{lem:tensor-effects}:
  \[
    \quant{\forall\ \Psi \in c_{\bm A \otimes \bm B}}{\input{bi-braket-fg1.tikz}\ =\ 1}
    \ \ \iff\ \ 
    \quant{\forall\ \psi \in c_{\bm A}, \phi \in c_{\bm B}}{\input{bi-braket-fg2.tikz}\ =\ 1} 
  \]
  The RHS holds since $f \circ \psi \in c_{\bm A'}$ and $g \circ \phi \in c_{\bm B'}$. From the LHS above, we can conclude that $f \otimes g : \bm A \otimes \bm B \to \bm A' \otimes \bm B'$ is a morphism in \CausC.

  Finally, it remains to show that the swap is a morphism in \CausC. By Proposition \ref{prop:morphism}, this is the case when, for all $\pi \in c_{\bm B \otimes \bm A}^*$, we have:
  \[ \input{swap-effect.tikz}\ \in\ c_{\bm A \otimes \bm B}^* \]
  This again follows by relying on Lemma \ref{lem:tensor-effects}.
\end{proof}

\begin{lemma*}
  The transposition functor $(-)^* : \mathcal C^{\textrm{op}} \to \mathcal C$:
  \begin{equation}\label{eq:trans-functor-app}
    A\mapsto A^* \qquad\qquad \input{transpose.tikz}
  \end{equation}
  lifts to a full and faithful functor $(-)^* : \CausC^{\textrm{op}} \to \CausC$, where $\bm A^* := (A^*, c_{\bm A^*})$.
\end{lemma*}
\begin{proof}
  Note that $c_{\bm B}^* = c_{\bm B^*}$, by definition, and $c_{\bm A} = c_{\bm A}^{**} = (c_{\bm A^*})^*$. Hence, for $f : \bm A \to \bm B$ we have:
  \[ \quant{\ \forall\ \rho \in c_{\bm A}, \pi \in c_{\bm B}^*\ }{\input{dual-f1.tikz}\ =\ 1}
    \ \ \iff\ \ 
    \quant{\ \forall\ \pi \in c_{\bm B^*}, \rho \in (c_{\bm A^*})^*\ }{\input{dual-f2.tikz}\ =\ 1}  \]
  so $f^* : \bm B^* \to \bm A^*$ is a morphism in \CausC.
  Just as with the functor $(-)^*$ in $\mathcal C$, $((-)^*)^* = \textrm{Id}_{\CausC}$, so fullness and faithfulness is immediate.
\end{proof}

\begin{theorem*}
  For any precausal category $\mathcal C$, \CausC is a $*$-autonomous category where $\bm I = \bm I^*$.
\end{theorem*}
\begin{proof}
  We have already shown that \CausC is an SMC (Theorem~\ref{CausCisSMC}) with a full and faithful functor
   $(-)^* : \CausC^{\textrm{op}} \to \CausC$ (Lemma~\ref{lem:star}).
    Consider objects $\bm A, \bm B, \bm C$ in \CausC. The underlying object of $\bm B \limp \bm C$ is:
  \[ (B \otimes C^*)^* = B^{*} \otimes C^{**} = B^* \otimes C \]
  Since $\mathcal C$ is compact closed, there is a natural isomorphism:
  \[  \mathcal C(A \otimes B, C) \cong \mathcal C(A, B^* \otimes C) \]
  given by:
  \[ \input{fABtoC.tikz} \quad \mapsto \quad
     \input{fAtoBstarC.tikz}\ \ =:\ \ \input{gAtoBstarC.tikz} \]
  whose inverse is:
  \[ \input{gAtoBstarC.tikz} \quad \mapsto \quad
     \input{gABtoC.tikz} =\ \ \input{fABtoC.tikz} \]
  Indeed this is how one shows that compact closed categories are in fact closed. 
Thus, it suffices to show that:
  \[ f \in \CausC(\bm A \otimes \bm B, \bm C) \iff g \in \CausC(\bm A, \bm B \limp \bm C) \]
  This follows from Lemma \eqref{lem:tensor-effects}:
\begin{align*}
    f \in \CausC(\bm A \otimes \bm B, \bm C)
    &\ \iff \ 
    \quant{\forall \rho \in c_{\bm A \otimes \bm B}, \pi \in c_{\bm C}^*}
    {\input{fABtoC-sandwich1.tikz}\ =\ 1} \\[3mm]
    & \iff \ \quant{\forall \rho_1 \in c_{\bm A}, \rho_2 \in c_{\bm B}, \pi \in c_{\bm C}^*}
    {\input{fABtoC-sandwich2.tikz}\ =\ 1} \\[3mm]
    &\ \iff \ 
    \quant{\forall \rho_1 \in c_{\bm A}, \rho_2 \in c_{\bm B^*}^*, \pi \in c_{\bm C}^*}
    {\input{fABtoC-sandwich3.tikz}\ =\ 1} \\[3mm]
    &\  \iff \ 
    g \in \CausC(\bm A, \bm B \limp \bm C)
  \end{align*}
  Finally, $\bm I = \bm I^*$ follows from the fact that $I = I^*$ and
  \[
    c^*_{\bm I} = \{\lambda \mid 1 \lambda = 1\} = \{1\} = c_{\bm I}
    \tag*{\qedhere}
  \]
\end{proof}

\subsection{Proofs that \MatRp and \CPM are precausal}\label{sec:mat-cpm-precausal}

\begin{theorem*}
  \MatRp is a precausal category.
\end{theorem*}
\begin{proof}
  \Ca was given in Example \ref{ex:mat-causality}. \Cb is immediate, and \Cc follows from the fact that one can always construct a basis for a vector space out of probability distributions, e.g., by taking the point distributions. To show \Cd, we will decompose into \Ce and \Cf via Proposition~\ref{4to45}.

  So, we turn to \Ce:
  \[ \quant{\exists\ \Phi'\ \textit{causal}}{\input{one-way-sig.tikz}}
     \implies
     \quant{\exists\ \Phi_1,\Phi_2\ \textit{causal}}{\input{one-way-sig2.tikz}}
  \]
  In terms of a conditional probability distribution $P(A' B' | A B)$, the premise above amounts to the usual non-signalling condition:
  \[ P(A' | A B) = P(A' | A) \]
  Hence the conclusion follows from the product rule:
  \begin{align*}
    P(A' B'|A B) & = P(A' | A B) P(B' | A' A B) \\
    &  = P(A' | A) P(B' | A' A B)
  \end{align*}
  More precisely, suppose $\Phi_{ij}^{kl}$ is a stochastic matrix such that there exists another stochastic matrix $(\Phi')_i^k$ where:
  \[ \sum_l \Phi_{ij}^{kl} = (\Phi')_i^k \]
  Then, let:
  \begin{align*}
    (\Phi_1)_{i}^{ki'k'} & = (\Phi')_i^k \delta_{ii'} \delta_{kk'} \\
    (\Phi_2)_{i'k'j}^l & = \begin{cases}
      \delta_{0l} & \textrm{ if } (\Phi')_i^{k'} = 0 \\
      \Phi_{ij}^{k'l}/(\Phi')_i^{k'} & \textrm{ otherwise}
    \end{cases}
  \end{align*}
  where $\delta_{ij}$ is the Kronecker delta. One can straightforwardly verify that these are both stochastic matrices. Let $\Psi_{ij}^{kl}$ be the result of plugging outputs $i',k'$ of $\Phi_1$ into those inputs for $\Phi_2$, i.e.
  \[
  \Psi_{ij}^{kl}
  := \sum_{i'k'} (\Phi_1)_{i}^{ki'k'} (\Phi_2)_{i'k'j}^l
  = (\Phi')_i^k (\Phi_2)_{ikj}^l
  \]
  If $(\Phi')_i^k = 0$, then both $\Phi_{ij}^{kl}$ and $\Psi_{ij}^{kl}$ are $0$ for all $j, l$. So, suppose $(\Phi')_i^k \neq 0$. Then:
  \[ \Psi_{ij}^{kl} 
  = (\Phi')_i^k (\Phi_{ij}^{kl}/(\Phi')_i^k)
  = \Phi_{ij}^{kl} \]

  For \Cf, let $w_j^i$ be the matrix of a second-order causal effect $w : A \otimes B^*$. Then for all stochastic matrices $\Phi_i^j$, we have:
  \[ \sum_{ij} w_j^i \Phi_i^j = 1 \]
  For some fixed column $m$, and fixed rows $n \neq n'$, the following matrix:
  \[
  1 = \Phi_i^j = \begin{cases}
    p & i = m, j = n \\
    1 - p & i = m, j = n' \\
    0 & i = m, j \neq n, j \neq n' \\
    \delta_i^j & i \neq m
  \end{cases}
  \]
  defines a stochastic map for any $p \in [0,1]$. Then:
  \[ \sum_{ij} w_j^i \Phi_i^j = p w_n^m + (1-p) w_{n'}^m + K = 1 \]
  where $K$ doesn't depend on $p$. Since we can freely vary $p$ between $0$ and $1$, the only way to preserve normalisation is if $w_n^m = w_{n'}^m$. Hence, for all $j$, we have $w_j^i = w_0^i$. Defining $\rho^i := w_0^i$ gives factorisation \Cf.
\end{proof}

\begin{theorem*}
The category \CPM is a precausal category.
\end{theorem*}
\begin{proof}
As noted in Example~\ref{ex:cpm-causality}, discarding processes satisfying \Ca are given by the trace. For \Cb, $d_{\mathcal L(H)} = \textrm{dim}(H)^2$, which is invertible whenever $\textrm{dim}(H) \neq 0$. \Cc follows from the fact that for any finite-dimensional Hilbert space $H$, we can find a basis of positive operators spanning $\mathcal L(H)$. Hence, by renormalising, we can also find a basis of trace-$1$ positive operators.

For \Cd, we shall show \Ce and \Cf.

Condition \Ce states that
$$tr_{B'}(\Phi) = \Phi' \implies \Phi = (1_{A'} \otimes \Phi_2) \circ (\Phi_1 \otimes 1_B)$$
This is precisely the result of \cite{Eggeling} and is based on the fact that minimal Stinespring dilations are related by a unitary.


For \Cf, a causal map $\Phi : A\to B$ in \CPM is a completely positive trace preserving map.
Such a map can always be written as 
\ctikzfig{CPTP-in-basis}
where the states and effects labeled with $i$ ($j$) form a basis for $B$ ($A^*$), with the 0-th basis element the maximally mixed state (discard effect).
Now if any second order causal effect $w$ does not split as in \Cf, we can always change the value of some of the $r_{i,j}$ such that $\Phi$ is still positive, but $w(\Phi) \neq 1$.
\end{proof}

\subsection{Proofs that \Rel is only weakly pre-causal}

\begin{theorem*}
The category \Rel satisfies axioms \Ca-\Cc and \Ce.
\end{theorem*}

\begin{proof}
  It will be conventient to write relations as (possibly infinite) matrices taking entries the booleans. That is, a relation $f \subseteq A \times B$ is equivalently represented as the following boolean matrix:
  \[ \left\{ f_i^j \in \mathbb B | i \in A, j \in B \right\} \]
  Relation composition then becomes:
  \[ h_i^k := \sum_j f_i^j g_j^k \]
  where we interpret multiplication as meet and (possibly infinite) summation as join.

  For \Ca, the discarding map for any system $A$ is the relation which relates every $i \in A$ to $* \in I := \{*\}$. That is:
  \[ \discard_i = 1 \]
  This clearly respects the monoidal structure. Using this definition of discarding, causal relations are those $f$ such that for all $i \in A$ there exists $j \in B$ such that $f_i^j = 1$.

  The empty set is the zero object in \Rel. The definition of discarding yields that $d_A = 1$ for all non-empty sets $A$, hence \Cb is satisfied.

  \Cc is immediate consequence of the fact that all singletons $\{i\} \subseteq A$ are causal states in \Rel.

  \Ce can be proven in almost the same way as for \MatRp. That only difference is we no longer need to use division in the definiton of $\Phi_2$, because $0$ and $1$ are the only possible values that $(\Phi')_i^{k'}$ can take:
  \begin{align*}
    (\Phi_1)_{i}^{ki'k'} & = (\Phi')_i^k \delta_{ii'} \delta_{kk'} \\
    (\Phi_2)_{i'k'j}^l & = \begin{cases}
      \delta_{0l} & \textrm{ if } (\Phi')_i^{k'} = 0 \\
      \Phi_{ij}^{k'l} & \textrm{ if } (\Phi')_i^{k'} = 1
    \end{cases}
  \end{align*}
  These are causal, and by case distinction on $(\Phi')_i^{k'} \in \{0, 1\}$, we can see that:
  \[
    \sum_{i'k'} (\Phi_1)_{i}^{ki'k'} (\Phi_2)_{i'k'j}^l = \Phi_{ij}^{kl}
    \tag*{\qedhere}
  \]
\end{proof}

\begin{theorem*}
The category \Rel does not satify \Cf.
\end{theorem*}

\begin{proof}
  Causality means for all $i$, there exists $j$ such that $f_i^j = 1$. The dual condition for causality is obtained by reversing the quantifies. That is, for some relation $w$:
  \begin{equation}\label{eq:rel-soc}
    f \ \ \textit{causal} \ \ \implies \ \ \sum_{ij} f_i^j w_j^i = 1
  \end{equation}
  if and only if there exists some $j$ such that for all $i$, $w_j^i = 1$.

  \Cf requires that any such $w$ must be independent of $i$. That is, $w_j^i = \rho^i$. However, $w \subseteq \mathbb B \times \mathbb B$ defined by $w_j^i := \neg(i \wedge j)$ satisfies \eqref{eq:rel-soc}, but is not independent of $i$.
\end{proof}

\subsection{Proof of embedding into the type of all causal processes}

\begin{theorem*}[\ref{thm:fo-embed}]
  Any object $\bm X$ built inductively from first order systems $\bm A_1,
  \ldots, \bm A_m$, $\bm A_1', \ldots, \bm A_n'$ has a canonical embedding of the form:
  \[ e : \bm X \rightarrow (\bm A_1 \otimes \ldots \otimes \bm A_m \limp \bm A_1' \otimes \ldots \otimes \bm A_n') \]
  whose underlying $\mathcal C$-morphism is just a permutation of systems.
\end{theorem*}

\begin{proof}
  Given $\bm X$ built inductively from first-order types via the $*$-autonomous structure, we can replace $\bm A \limp \bm B$ with $\bm A^* \parr \bm B$ and push the $(-)^*$ inside as far as possible via application of the following isomorphisms from left-to-right:
  \[
  (\bm A \otimes \bm B)^* \cong \bm A^* \parr \bm B^*
  \qquad \textrm \qquad
  (\bm A \parr \bm B)^* \cong \bm A^* \otimes \bm B^*
  \]
  We can then apply $\bm A^{**} \cong \bm A$ to reduce $\bm X$ to an expression consisting of either first-order types or their duals, combined with $\otimes$ and $\parr$.
  For example, if $\bm X := \bm A_1 \limp (\bm A_1' \limp \bm A_2) \limp \bm A_2'$, we have:
  \begin{align*}
    \bm X & := \bm A_1 \limp (\bm A_1' \limp \bm A_2) \limp \bm A_2' \\
    & \cong \bm A_1^* \parr ((\bm A_1')^* \parr \bm A_2)^* \parr \bm A_2' \\
    & \cong \bm A_1^* \parr ((\bm A_1')^{**} \otimes \bm A_2^*) \parr \bm A_2' \\
    & \cong \bm A_1^* \parr (\bm A_1' \otimes \bm A_2^*) \parr \bm A_2'
  \end{align*}
  It was then noted in section \ref{sec:causC} that there is a canonical embedding $\bm A \otimes \bm B \hookrightarrow \bm A \parr \bm B$. We then apply this embedding to eliminate all $\otimes$'s, and then apply a permutation in $\mathcal C$ to sort all of the dual first-order systems to the left. Continuing the example, we have:
  \begin{align*}
    \bm X & \cong \bm A_1^* \parr (\bm A_1' \otimes \bm A_2^*) \parr \bm A_2' \\
    & \hookrightarrow \bm A_1^* \parr \bm A_1' \parr \bm A_2^* \parr \bm A_2' \\
    & \overset{(\dagger)}{\cong} \bm A_1^* \parr \bm A_2^* \parr \bm A_1' \parr \bm A_2'
  \end{align*}
  We can then pull the $(-)^*$ out as far as possible and re-introduce the $\limp$. Then, by Corollary~\ref{cor:fo-tensor-is-par}, we can always replace a $\parr$ of first-order systems with a $\otimes$, which completes the embedding. Applying these steps to the running example yields:
  \begin{align*}
    \bm X & \cong \bm A_1^* \parr \bm A_2^* \parr \bm A_1' \parr \bm A_2' \\
    & \cong (\bm A_1 \otimes \bm A_2)^* \parr \bm A_1' \parr \bm A_2' \\
    & \cong \bm A_1 \otimes \bm A_2 \limp \bm A_1' \parr \bm A_2' \\
    & \cong \bm A_1 \otimes \bm A_2 \limp \bm A_1' \otimes \bm A_2'
  \end{align*}
  We complete the proof by noting that the only step not arising from an identity morphism in $\mathcal C$ is the step marked $(\dagger)$ above, which is a permutation.
\end{proof}

\end{document}

%% file: preamble-lmcs.tex
\usepackage{amsmath,amsthm,amssymb}
\usepackage{stmaryrd,cmll}
\usepackage{proof,linear}

\usepackage{xspace,enumerate,color,epsfig}
\usepackage{graphicx}
\graphicspath{{.}{./figures/}}
\usepackage{marvosym}
\usepackage{bm}

\usepackage{tikzfig}
\usepackage{docmute}
\usepackage{keycommand}

\input{defs.tex}
\input{tikzstyles.tex}

\input{tikzfigures.tex}
\let\olddagger\dagger
\renewcommand{\dagger}{\ensuremath{\olddagger}\xspace}

\theoremstyle{definition}
\newtheorem{theorem}{Theorem}[section]
\newtheorem{corollary}[theorem]{Corollary}
\newtheorem{lemma}[theorem]{Lemma}
\newtheorem{proposition}[theorem]{Proposition}

\newtheorem{definition}[theorem]{Definition}
\newtheorem{example}[theorem]{Example}

\newtheorem{example*}[theorem]{Example*}
\newtheorem{examples*}[theorem]{Examples*}
\newtheorem{remark}[theorem]{Remark}
\newtheorem{remark*}[theorem]{Remark*}

\newtheorem*{theorem*}{Theorem}
\newtheorem*{corollary*}{Corollary}
\newtheorem*{lemma*}{Lemma}
\newtheorem*{proposition*}{Proposition}

\newcommand{\quant}[2]{\left(\begin{matrix}
\ \textrm{\footnotesize $#1$}\ .\  \\[1mm]
\ #2\ 
\end{matrix}\right)}

\newcommand{\lquant}[2]{\left(\textrm{\footnotesize $#1$}\ .\ \ #2\ \right)}

\newcommand{\Ca}{\ensuremath{(\textrm{C1})}\xspace}
\newcommand{\Cb}{\ensuremath{(\textrm{C2})}\xspace}
\newcommand{\Cc}{\ensuremath{(\textrm{C3})}\xspace}
\newcommand{\Cd}{\ensuremath{(\textrm{C4})}\xspace}
\newcommand{\Ce}{\ensuremath{(\textrm{C4}')}\xspace}
\newcommand{\Cf}{\ensuremath{(\textrm{C5}')}\xspace}

\newcommand\Caus{\textrm{Caus}}
\newcommand\CausC{\ensuremath{\Caus[\mathcal C]}\xspace}

\usepackage[color]{changebar}

\usepackage{color}
\def\bR{\begin{color}{red}} 
\def\bB{\begin{color}{blue}}
\def\bM{\begin{color}{magenta}}
\def\bC{\begin{color}{cyan}}
\def\bW{\begin{color}{white}}
\def\bBl{\begin{color}{black}} 
\def\bG{\begin{color}{green}}
\def\bY{\begin{color}{yellow}}
\def\e{\end{color}\xspace}
\newcommand{\bit}{\begin{itemize}}
\newcommand{\eit}{\end{itemize}\par\noindent}
\newcommand{\ben}{\begin{enumerate}}
\newcommand{\een}{\end{enumerate}\par\noindent}
\newcommand{\beq}{\begin{equation}}
\newcommand{\eeq}{\end{equation}\par\noindent}
\newcommand{\beqa}{\begin{eqnarray*}}
\newcommand{\eeqa}{\end{eqnarray*}\par\noindent}
\newcommand{\beqn}{\begin{eqnarray}}
\newcommand{\eeqn}{\end{eqnarray}\par\noindent}



%% file: defs.tex
\newcommand{\emptydiag}{\,\tikz{\node[style=empty diagram] (x) {};}\,}

\newkeycommand{\boxmap}[style=small box][1]{\,\tikz{\node[style=\commandkey{style}] (x) {$#1$};\draw(0,-1.25)--(x)--(0,1.25);}\,}
\newkeycommand{\dboxmap}[style=dbox][1]{\,\tikz{\node[style=\commandkey{style}] (x) {$#1$};\draw[boldedge] (0,-1.25)--(x)--(0,1.25);}\,}

\newkeycommand{\wirelabel}[style=white label]{\,\tikz{\node[style=\commandkey{style}]{};}\,\xspace}

\newkeycommand{\pointmap}[style=point][1]{\,\tikz{\node[style=\commandkey{style}] (x) at (0,-0.3) {$#1$};\node [style=none] (3) at (0, 0.9) {};\node [style=none] (3) at (0, -0.9) {};\draw (x)--(0,0.7);}\,}

\newkeycommand{\point}[style=point,estyle=][1]{\,\tikz{\node[style=\commandkey{style}] (x) at (0,-0.05) {$#1$};\draw [\commandkey{estyle}] (x)--(0,1.05);}\,}

\newkeycommand{\copointmap}[style=copoint][1]{\,\tikz{\node[style=\commandkey{style}] (x) at (0,0.3) {$#1$};\draw (x)--(0,-0.7);}\,}

\newkeycommand{\dpoint}[style=point][1]{\,\tikz{\node[style=\commandkey{style},doubled] (x) at (0,-0.3) {$#1$};\draw[boldedge] (x)--(0,0.7);}\,}

\newkeycommand{\kpoint}[style=kpoint,estyle=][1]{\,\tikz{\node[style=\commandkey{style}] (x) at (0,-0.05) {$#1$};\draw [\commandkey{estyle}] (x)--(0,1.05);}\,}

\newkeycommand{\kpointgray}[style=gray kpoint,estyle=][1]{\,\tikz{\node[style=\commandkey{style}] (x) at (0,-0.05) {$#1$};\draw [\commandkey{estyle}] (x)--(0,1.05);}\,}

\newkeycommand{\typedkpoint}[style=kpoint,edgestyle=][2]{%
\,\begin{tikzpicture}
  \begin{pgfonlayer}{nodelayer}
    \node [style=none] (0) at (0, 1) {};
    \node [style=\commandkey{style}] (1) at (0, -0.75) {$#2$};
    \node [style=right label] (2) at (0.25, 0.5) {$#1$};
  \end{pgfonlayer}
  \begin{pgfonlayer}{edgelayer}
    \draw [\commandkey{edgestyle}] (1) to (0.center);
  \end{pgfonlayer}
\end{tikzpicture}\,}

\newkeycommand{\typedkpointdag}[style=kpoint adjoint,edgestyle=][2]{%
\,\begin{tikzpicture}
  \begin{pgfonlayer}{nodelayer}
    \node [style=none] (0) at (0, -1) {};
    \node [style=\commandkey{style}] (1) at (0, 0.75) {$#2$};
    \node [style=right label] (2) at (0.25, -0.5) {$#1$};
  \end{pgfonlayer}
  \begin{pgfonlayer}{edgelayer}
    \draw [\commandkey{edgestyle}] (0.center) to (1);
  \end{pgfonlayer}
\end{tikzpicture}\,}

\newkeycommand{\bistate}[style=kpoint][1]{\,\begin{tikzpicture}
  \begin{pgfonlayer}{nodelayer}
    \node [style=\commandkey{style}, minimum width=1 cm, inner sep=2pt] (0) at (0, -0.25) {$#1$};
    \node [style=none] (1) at (-0.75, 0) {};
    \node [style=none] (2) at (0.75, 0) {};
    \node [style=none] (3) at (-0.75, 0.88) {};
    \node [style=none] (4) at (0.75, 0.88) {};
  \end{pgfonlayer}
  \begin{pgfonlayer}{edgelayer}
    \draw (1.center) to (3.center);
    \draw (2.center) to (4.center);
  \end{pgfonlayer}
\end{tikzpicture}\,}

\newkeycommand{\bistatebig}[style=kpoint][1]{\,\begin{tikzpicture}
  \begin{pgfonlayer}{nodelayer}
    \node [style=\commandkey{style}, minimum width=, minimum width=1.26 cm, inner sep=2pt] (0) at (0, -0.25) {$#1$};
    \node [style=none] (1) at (-1, 0) {};
    \node [style=none] (2) at (1, 0) {};
    \node [style=none] (3) at (-1, 0.88) {};
    \node [style=none] (4) at (1, 0.88) {};
  \end{pgfonlayer}
  \begin{pgfonlayer}{edgelayer}
    \draw (1.center) to (3.center);
    \draw (2.center) to (4.center);
  \end{pgfonlayer}
\end{tikzpicture}\,}

\newkeycommand{\dbistate}[style=dkpoint][1]{\,\begin{tikzpicture}
  \begin{pgfonlayer}{nodelayer}
    \node [style=\commandkey{style}, minimum width=1 cm, inner sep=2pt] (0) at (0, -0.25) {$#1$};
    \node [style=none] (1) at (-0.75, 0) {};
    \node [style=none] (2) at (0.75, 0) {};
    \node [style=none] (3) at (-0.75, 1.0) {};
    \node [style=none] (4) at (0.75, 1.0) {};
  \end{pgfonlayer}
  \begin{pgfonlayer}{edgelayer}
    \draw [boldedge] (1.center) to (3.center);
    \draw [boldedge] (2.center) to (4.center);
  \end{pgfonlayer}
\end{tikzpicture}\,}

\newkeycommand{\bistatebraket}[style1=kpoint,style2=kpointadj][2]{\,\begin{tikzpicture}
  \begin{pgfonlayer}{nodelayer}
    \node [style=\commandkey{style1}, minimum width=1 cm, inner sep=2pt] (0) at (0, -0.75) {$#1$};
    \node [style=none] (1) at (-0.75, -0.5) {};
    \node [style=none] (2) at (0.75, -0.5) {};
    \node [style=none] (3) at (-0.75, 0.5) {};
    \node [style=none] (4) at (0.75, 0.5) {};
    \node [style=\commandkey{style2}, minimum width=1 cm, inner sep=2pt] (5) at (0, 0.75) {$#2$};
  \end{pgfonlayer}
  \begin{pgfonlayer}{edgelayer}
    \draw (1.center) to (3.center);
    \draw (2.center) to (4.center);
  \end{pgfonlayer}
\end{tikzpicture}\,}

\newcommand{\weight}[1]{\,\begin{tikzpicture}
  \begin{pgfonlayer}{nodelayer}
    \node [style=point] (0) at (0, -0.25) {$#1$};
    \node [style=upground] (1) at (0, 0.75) {};
  \end{pgfonlayer}
  \begin{pgfonlayer}{edgelayer}
    \draw (0) to (1);
  \end{pgfonlayer}
\end{tikzpicture}\,}

\newkeycommand{\dkpoint}[style=dkpoint][1]{\,\tikz{\node[style=\commandkey{style}] (x) at (0,-0.1) {$#1$};\draw[boldedge] (x)--(0,1);}\,}
\newkeycommand{\dkpointadj}[style=dkpointadj][1]{\,\tikz{\node[style=\commandkey{style}] (x) at (0,0.1) {$#1$};\draw[boldedge] (x)--(0,-1);}\,}
\newkeycommand{\dkpointtrans}[style=dkpointtrans][1]{\,\tikz{\node[style=\commandkey{style}] (x) at (0,0.1) {$#1$};\draw[boldedge] (x)--(0,-1);}\,}
\newkeycommand{\dkpointconj}[style=dkpointconj][1]{\,\tikz{\node[style=\commandkey{style}] (x) at (0,-0.1) {$#1$};\draw[boldedge] (x)--(0,1);}\,}

\newcommand{\trace}{\,\begin{tikzpicture}
  \begin{pgfonlayer}{nodelayer}
    \node [style=none] (0) at (0, -0.60) {};
    \node [style=upground] (1) at (0, 0.40) {};
  \end{pgfonlayer}
  \begin{pgfonlayer}{edgelayer}
    \draw (0.center) to (1);
  \end{pgfonlayer}
\end{tikzpicture}\,}

\newcommand\discard\trace

\newcommand\disc{\scalebox{0.7}{\,\begin{tikzpicture}
  \begin{pgfonlayer}{nodelayer}
    \node [style=none] (0) at (0, -0.4) {};
    \node [style=upground] (1) at (0, 0.4) {};
  \end{pgfonlayer}
  \begin{pgfonlayer}{edgelayer}
    \draw (0.center) to (1);
  \end{pgfonlayer}
\end{tikzpicture}\,}}

\newcommand{\maxmix}{\,\begin{tikzpicture}
  \begin{pgfonlayer}{nodelayer}
    \node [style=none] (0) at (-0.25, 0.60) {};
    \node [style=downground] (1) at (-0.25, -0.40) {};
  \end{pgfonlayer}
  \begin{pgfonlayer}{edgelayer}
    \draw (0.center) to (1);
  \end{pgfonlayer}
\end{tikzpicture}\,}

\newcommand{\mmix}{\scalebox{0.7}{\maxmix}}

\newkeycommand{\pointketbra}[style1=copoint,style2=point][2]{\,%
\begin{tikzpicture}
  \begin{pgfonlayer}{nodelayer}
    \node [style=none] (0) at (0, -1.75) {};
    \node [style=\commandkey{style1}] (1) at (0, -0.75) {$#1$};
    \node [style=\commandkey{style2}] (2) at (0, 0.75) {$#2$};
    \node [style=none] (3) at (0, 1.75) {};
  \end{pgfonlayer}
  \begin{pgfonlayer}{edgelayer}
    \draw (0.center) to (1);
    \draw (2) to (3.center);
  \end{pgfonlayer}
\end{tikzpicture}\,}

\newkeycommand{\twocopointketbra}[style1=copoint,style2=copoint,style3=point][3]{\,%
\begin{tikzpicture}
  \begin{pgfonlayer}{nodelayer}
    \node [style=none] (0) at (-0.75, -1.75) {};
    \node [style=none] (1) at (0.75, -1.75) {};
    \node [style=\commandkey{style1}] (2) at (-0.75, -0.75) {$#1$};
    \node [style=\commandkey{style2}] (3) at (0.75, -0.75) {$#2$};
    \node [style=\commandkey{style3}] (4) at (0, 0.75) {$#3$};
    \node [style=none] (5) at (0, 1.75) {};
  \end{pgfonlayer}
  \begin{pgfonlayer}{edgelayer}
    \draw (0.center) to (2);
    \draw (1.center) to (3);
    \draw (4) to (5.center);
  \end{pgfonlayer}
\end{tikzpicture}\,}

\newkeycommand{\twopointketbra}[style1=copoint,style2=point,style3=point][3]{\,%
\begin{tikzpicture}
  \begin{pgfonlayer}{nodelayer}
    \node [style=none] (0) at (0, -1.75) {};
    \node [style=none] (1) at (-0.75, 1.75) {};
    \node [style=\commandkey{style1}] (2) at (0, -0.75) {$#1$};
    \node [style=\commandkey{style2}] (3) at (-0.75, 0.75) {$#2$};
    \node [style=\commandkey{style3}] (4) at (0.75, 0.75) {$#3$};
    \node [style=none] (5) at (0.75, 1.75) {};
  \end{pgfonlayer}
  \begin{pgfonlayer}{edgelayer}
    \draw (0.center) to (2);
    \draw (1.center) to (3);
    \draw (4) to (5.center);
  \end{pgfonlayer}
\end{tikzpicture}\,}

\newkeycommand{\pointbraket}[style1=point,style2=copoint,estyle=][2]{\,%
\begin{tikzpicture}
  \begin{pgfonlayer}{nodelayer}
    \node [style=\commandkey{style1}] (0) at (0, -0.625) {$#1$};
    \node [style=\commandkey{style2}] (1) at (0, 0.625) {$#2$};
  \end{pgfonlayer}
  \begin{pgfonlayer}{edgelayer}
    \draw [\commandkey{estyle}] (0) to (1);
  \end{pgfonlayer}
\end{tikzpicture}\,}

\newkeycommand{\kpointketbra}[style1=kpointdag,style2=kpoint,estyle=][2]{\,%
\begin{tikzpicture}
  \begin{pgfonlayer}{nodelayer}
    \node [style=none] (0) at (0, -1.9) {};
    \node [style=\commandkey{style1}] (1) at (0, -0.95) {$#1$};
    \node [style=\commandkey{style2}] (2) at (0, 0.95) {$#2$};
    \node [style=none] (3) at (0, 1.9) {};
  \end{pgfonlayer}
  \begin{pgfonlayer}{edgelayer}
    \draw [\commandkey{estyle}] (0.center) to (1);
    \draw [\commandkey{estyle}] (2) to (3.center);
  \end{pgfonlayer}
\end{tikzpicture}\,}

\newkeycommand{\kpointmap}[style1=kpoint,style2=map][2]{\,%
\begin{tikzpicture}
  \begin{pgfonlayer}{nodelayer}
    \node [style=\commandkey{style1}] (0) at (0, -1.1) {$#1$};
    \node [style=\commandkey{style2}] (1) at (0, 0.5) {$#2$};
    \node [style=none] (2) at (0, 1.75) {};
  \end{pgfonlayer}
  \begin{pgfonlayer}{edgelayer}
    \draw (0) to (1);
    \draw (1) to (2.center);
  \end{pgfonlayer}
\end{tikzpicture}%
\,}

\newkeycommand{\dkpointmap}[style1=dkpoint,style2=dmap][2]{\,%
\begin{tikzpicture}
  \begin{pgfonlayer}{nodelayer}
    \node [style=\commandkey{style1}] (0) at (0, -1.2) {$#1$};
    \node [style=\commandkey{style2}] (1) at (0, 0.4) {$#2$};
    \node [style=none] (2) at (0, 1.7) {};
  \end{pgfonlayer}
  \begin{pgfonlayer}{edgelayer}
    \draw[style=boldedge] (0) to (1);
    \draw[style=boldedge] (1) to (2.center);
  \end{pgfonlayer}
\end{tikzpicture}%
\,}

\newkeycommand{\dkpointmapx}[style1=dkpoint,style2=dmap][2]{\,%
\begin{tikzpicture}
  \begin{pgfonlayer}{nodelayer}
    \node [style=\commandkey{style1}] (0) at (0, -1.4) {$#1$};
    \node [style=\commandkey{style2}] (1) at (0, 0.6) {$#2$};
    \node [style=none] (2) at (0, 1.9) {};
  \end{pgfonlayer}
  \begin{pgfonlayer}{edgelayer}
    \draw[style=boldedge] (0) to (1);
    \draw[style=boldedge] (1) to (2.center);
  \end{pgfonlayer}
\end{tikzpicture}%
\,}

\newkeycommand{\kpointsandwichmap}[style1=kpoint,style2=map,style3=kpointdag][3]{\,%
\begin{tikzpicture}
  \begin{pgfonlayer}{nodelayer}
    \node [style=\commandkey{style1}] (0) at (0, -1.5) {$#1$};
    \node [style=\commandkey{style2}] (1) at (0, 0) {$#2$};
    \node [style=\commandkey{style3}] (2) at (0, 1.5) {$#3$};
  \end{pgfonlayer}
  \begin{pgfonlayer}{edgelayer}
    \draw (0) to (1);
    \draw (1) to (2);
  \end{pgfonlayer}
\end{tikzpicture}%
}

\newkeycommand{\sandwichtwo}[style1=point,style2=map,style3=map,style4=copoint][4]{\,%
\begin{tikzpicture}
  \begin{pgfonlayer}{nodelayer}
    \node [style=\commandkey{style2}] (0) at (0, -1) {$#2$};
    \node [style=\commandkey{style3}] (1) at (0, 1) {$#3$};
    \node [style=\commandkey{style1}] (2) at (0, -2.6) {$#1$};
    \node [style=\commandkey{style4}] (3) at (0, 2.6) {$#4$};
  \end{pgfonlayer}
  \begin{pgfonlayer}{edgelayer}
    \draw (2) to (0);
    \draw (0) to (1);
    \draw (1) to (3);
  \end{pgfonlayer}
\end{tikzpicture}\,}

\newkeycommand{\longbraket}[style1=point,style2=copoint][2]{\,%
\begin{tikzpicture}
  \begin{pgfonlayer}{nodelayer}
    \node [style=\commandkey{style1}] (0) at (0, -2.6) {$#1$};
    \node [style=\commandkey{style2}] (1) at (0, 2.6) {$#2$};
  \end{pgfonlayer}
  \begin{pgfonlayer}{edgelayer}
    \draw (0) to (1);
  \end{pgfonlayer}
\end{tikzpicture}\,}

\newkeycommand{\onb}[style=point]{\ensuremath{\{\,\tikz{\node[style=\commandkey{style}] (x) at (0,-0.3) {$j$};\draw (x)--(0,0.7);}\,\}}\xspace}

\newkeycommand{\phase}[style=white phase dot][1]{\,\begin{tikzpicture}
    \begin{pgfonlayer}{nodelayer}
        \node [style=none] (0) at (0, 1) {};
        \node [style=\commandkey{style}] (2) at (0, -0) {$#1$}; 
        \node [style=none] (3) at (0, -1) {};
    \end{pgfonlayer}
    \begin{pgfonlayer}{edgelayer}
        \draw (2) to (0.center);
        \draw (3.center) to (2);
    \end{pgfonlayer}
\end{tikzpicture}\,}

\newkeycommand{\dphase}[style=white phase ddot][1]{\,\begin{tikzpicture}
    \begin{pgfonlayer}{nodelayer}
        \node [style=none] (0) at (0, 1.25) {};
        \node [style=\commandkey{style}] (2) at (0, -0) {$#1$};
        \node [style=none] (3) at (0, -1.25) {};
    \end{pgfonlayer}
    \begin{pgfonlayer}{edgelayer}
        \draw [boldedge] (2) to (0.center);
        \draw [boldedge] (3.center) to (2);
    \end{pgfonlayer}
\end{tikzpicture}\,}

\newkeycommand{\dphasepoint}[style=white phase ddot][1]{\,\begin{tikzpicture}[yshift=-3mm]
  \begin{pgfonlayer}{nodelayer}
    \node [style=none] (0) at (0, 1) {};
    \node [style=\commandkey{style}] (1) at (0, 0) {$#1$};
  \end{pgfonlayer}
  \begin{pgfonlayer}{edgelayer}
    \draw [style=boldedge] (0.center) to (1);
  \end{pgfonlayer}
\end{tikzpicture}\,}

\newkeycommand{\dphasepointgray}[style=gray phase ddot][1]{\,\begin{tikzpicture}[yshift=-3mm]
  \begin{pgfonlayer}{nodelayer}
    \node [style=none] (0) at (0, 1) {};
    \node [style=\commandkey{style}] (1) at (0, 0) {$#1$};
  \end{pgfonlayer}
  \begin{pgfonlayer}{edgelayer}
    \draw [style=boldedge] (0.center) to (1);
  \end{pgfonlayer}
\end{tikzpicture}\,}

\newkeycommand{\dphasecopoint}[style=white phase ddot][1]{\,\begin{tikzpicture}[yshift=3mm,yscale=-1]
  \begin{pgfonlayer}{nodelayer}
    \node [style=none] (0) at (0, 1) {};
    \node [style=\commandkey{style}] (1) at (0, 0) {$#1$};
  \end{pgfonlayer}
  \begin{pgfonlayer}{edgelayer}
    \draw [style=boldedge] (0.center) to (1);
  \end{pgfonlayer}
\end{tikzpicture}\,}

\newkeycommand{\dphasepointsm}[style=white phase ddot][1]{\,\begin{tikzpicture}[yshift=-3mm]
  \begin{pgfonlayer}{nodelayer}
    \node [style=none] (0) at (0, 1) {};
    \node [style=\commandkey{style}] (1) at (0, 0) {\footnotesize\!$#1$\!};
  \end{pgfonlayer}
  \begin{pgfonlayer}{edgelayer}
    \draw [style=boldedge] (0.center) to (1);
  \end{pgfonlayer}
\end{tikzpicture}\,}

\newkeycommand{\phasepoint}[style=white phase dot][1]{\,\begin{tikzpicture}[yshift=-3mm]
  \begin{pgfonlayer}{nodelayer}
    \node [style=none] (0) at (0, 1) {};
    \node [style=\commandkey{style}] (1) at (0, 0) {$#1$};
  \end{pgfonlayer}
  \begin{pgfonlayer}{edgelayer}
    \draw (0.center) to (1);
  \end{pgfonlayer}
\end{tikzpicture}\,}

\newkeycommand{\phasecopoint}[style=white phase dot][1]{\,\begin{tikzpicture}[yshift=3mm]
  \begin{pgfonlayer}{nodelayer}
    \node [style=none] (0) at (0, -1) {};
    \node [style=\commandkey{style}] (1) at (0, 0) {$#1$};
  \end{pgfonlayer}
  \begin{pgfonlayer}{edgelayer}
    \draw (0.center) to (1);
  \end{pgfonlayer}
\end{tikzpicture}\,}

\newkeycommand{\kpointbraket}[style1=kpoint,style2=kpoint adjoint,estyle=][2]{\,%
\begin{tikzpicture}
  \begin{pgfonlayer}{nodelayer}
    \node [style=\commandkey{style1}] (0) at (0, -0.735) {$#1$};
    \node [style=\commandkey{style2}] (1) at (0, 0.735) {$#2$};
  \end{pgfonlayer}
  \begin{pgfonlayer}{edgelayer}
    \draw [\commandkey{estyle}] (0) to (1);
  \end{pgfonlayer}
\end{tikzpicture}\,}

\newkeycommand{\kpointbraketx}[style1=kpoint,style2=kpoint adjoint,estyle=][2]{\,%
\begin{tikzpicture}
  \begin{pgfonlayer}{nodelayer}
    \node [style=\commandkey{style1}] (0) at (0, -0.885) {$#1$};
    \node [style=\commandkey{style2}] (1) at (0, 0.885) {$#2$};
  \end{pgfonlayer}
  \begin{pgfonlayer}{edgelayer}
    \draw [\commandkey{estyle}] (0) to (1);
  \end{pgfonlayer}
\end{tikzpicture}\,}

\newcommand{\bra}[1]{\ensuremath{\left\langle #1 \right|}}
\newcommand{\ket}[1]{\ensuremath{\left|  #1 \right\rangle}}

\newcommand{\ketbra}[2]{\ensuremath{\ket{#1}\!\bra{#2}}}

\tikzstyle{cdiag}=[matrix of math nodes, row sep=3em, column sep=3em, text height=1.5ex, text depth=0.25ex,inner sep=0.5em]
\tikzstyle{arrow above}=[transform canvas={yshift=0.5ex}]
\tikzstyle{arrow below}=[transform canvas={yshift=-0.5ex}]



%% file: tikzstyles.tex
\tikzstyle{env}=[copoint,regular polygon rotate=0,minimum width=0.2cm, fill=black]

\tikzstyle{probs}=[shape=semicircle,fill=white,draw=black,shape border rotate=180,minimum width=1.2cm]

\tikzstyle{nudge}=[yshift=0.6mm]

\tikzstyle{every picture}=[baseline=-0.25em,scale=0.5]
\tikzstyle{dotpic}=[] 
\tikzstyle{diredges}=[every to/.style={diredge}]
\tikzstyle{math matrix}=[matrix of math nodes,left delimiter=(,right delimiter=),inner sep=2pt,column sep=1em,row sep=0.5em,nodes={inner sep=0pt},text height=1.5ex, text depth=0.25ex]

\tikzstyle{inline text}=[text height=1.5ex, text depth=0.25ex,yshift=0.5mm]
\tikzstyle{label}=[font=\footnotesize,text height=1.5ex, text depth=0.25ex]
\tikzstyle{left label}=[label,anchor=east,xshift=2mm]
\tikzstyle{right label}=[label,anchor=west,xshift=-2mm]

\newcommand{\phantombox}[1]{\tikz[baseline=(current bounding box).east]{\path [use as bounding box] (0,0) rectangle #1;}}
\tikzstyle{braceedge}=[decorate,decoration={brace,amplitude=2mm,raise=-1mm}]
\tikzstyle{small braceedge}=[decorate,decoration={brace,amplitude=1mm,raise=-1mm}]

\tikzstyle{doubled}=[line width=1.6pt] 
\tikzstyle{boldedge}=[doubled,shorten <=-0.17mm,shorten >=-0.17mm]
\tikzstyle{boldedgegray}=[doubled,gray,shorten <=-0.17mm,shorten >=-0.17mm]
\tikzstyle{singleedgegray}=[gray]

\tikzstyle{semidoubled}=[line width=1.4pt] 
\tikzstyle{semiboldedgegray}=[semidoubled,gray,shorten <=-0.17mm,shorten >=-0.17mm]

\tikzstyle{boxedge}=[semiboldedgegray]

\tikzstyle{boldedgedashed}=[very thick,dashed,shorten <=-0.17mm,shorten >=-0.17mm]
\tikzstyle{vboldedgedashed}=[doubled,dashed,shorten <=-0.17mm,shorten >=-0.17mm]
\tikzstyle{left hook arrow}=[left hook-latex]
\tikzstyle{right hook arrow}=[right hook-latex]
\tikzstyle{sembracket}=[line width=0.5pt,shorten <=-0.07mm,shorten >=-0.07mm]

\tikzstyle{causal edge}=[->,thick,gray]
\tikzstyle{causal nondir}=[thick,gray]
\tikzstyle{timeline}=[thick,gray, dashed]

\tikzstyle{cedge}=[<->,thick,gray!70!white]

\tikzstyle{empty diagram}=[draw=gray!40!white,dashed,shape=rectangle,minimum width=1cm,minimum height=1cm]
\tikzstyle{empty diagram small}=[draw=gray!50!white,dashed,shape=rectangle,minimum width=0.6cm,minimum height=0.5cm]

\tikzstyle{dot}=[inner sep=0mm,minimum width=2mm,minimum height=2mm,draw,shape=circle]  
\tikzstyle{Wsquare}=[white dot, shape=regular polygon, rounded corners=0.8 mm, minimum size=3.3 mm, regular polygon sides=3, outer sep=-0.2mm]
\tikzstyle{Wsquareadj}=[white dot, shape=regular polygon, rounded corners=0.8 mm, minimum size=3.3 mm, regular polygon sides=3, outer sep=-0.2mm, regular polygon rotate=180]
\tikzstyle{ddot}=[inner sep=0mm, doubled, minimum width=2.5mm,minimum height=2.5mm,draw,shape=circle]

\tikzstyle{black dot}=[dot,fill=black]
\tikzstyle{white dot}=[dot,fill=white,,text depth=-0.2mm]
\tikzstyle{white Wsquare}=[Wsquare,fill=white,,text depth=-0.2mm]
\tikzstyle{white Wsquareadj}=[Wsquareadj,fill=white,,text depth=-0.2mm]
\tikzstyle{green dot}=[white dot] 
\tikzstyle{gray dot}=[dot,fill=gray!40!white,,text depth=-0.2mm]
\tikzstyle{red dot}=[gray dot] 

\tikzstyle{black ddot}=[ddot,fill=black]
\tikzstyle{white ddot}=[ddot,fill=white]
\tikzstyle{gray ddot}=[ddot,fill=gray!40!white]

\tikzstyle{gray edge}=[gray!60!white]

\tikzstyle{small dot}=[inner sep=0.5mm,minimum width=0pt,minimum height=0pt,draw,shape=circle]

\tikzstyle{small black dot}=[small dot,fill=black]
\tikzstyle{small white dot}=[small dot,fill=white]
\tikzstyle{small gray dot}=[small dot,fill=gray!40!white]

\tikzstyle{causal dot}=[inner sep=0.4mm,minimum width=0pt,minimum height=0pt,draw=white,shape=circle,fill=gray!40!white]

\tikzstyle{phase dimensions}=[minimum size=5mm,font=\footnotesize,rectangle,rounded corners=2.5mm,inner sep=0.2mm,outer sep=-2mm]
\tikzstyle{dphase dimensions}=[minimum size=5mm,font=\footnotesize,rectangle,rounded corners=2.5mm,inner sep=0.2mm,outer sep=-2mm]

\tikzstyle{white phase dot}=[dot,fill=white,phase dimensions]
\tikzstyle{white phase ddot}=[ddot,fill=white,dphase dimensions]

\tikzstyle{white rect ddot}=[draw=black,fill=white,doubled,minimum size=5mm,font=\footnotesize,rectangle,rounded corners=2.5mm,inner sep=0.2mm]
\tikzstyle{gray rect ddot}=[draw=black,fill=gray!40!white,doubled,minimum size=6mm,font=\footnotesize,rectangle,rounded corners=3mm]

\tikzstyle{gray phase dot}=[dot,fill=gray!40!white,phase dimensions]
\tikzstyle{gray phase ddot}=[ddot,fill=gray!40!white,dphase dimensions]
\tikzstyle{grey phase dot}=[gray phase dot]
\tikzstyle{grey phase ddot}=[gray phase ddot]

\tikzstyle{small phase dimensions}=[minimum size=4mm,font=\tiny,rectangle,rounded corners=2mm,inner sep=0.2mm,outer sep=-2mm]
\tikzstyle{small dphase dimensions}=[minimum size=4mm,font=\tiny,rectangle,rounded corners=2mm,inner sep=0.2mm,outer sep=-2mm]

\tikzstyle{small gray phase dot}=[dot,fill=gray!40!white,small phase dimensions]
\tikzstyle{small gray phase ddot}=[ddot,fill=gray!40!white,small dphase dimensions]

\tikzstyle{small map}=[draw,shape=rectangle,minimum height=4mm,minimum width=4mm,fill=white]

\tikzstyle{cnot}=[fill=white,shape=circle,inner sep=-1.4pt]

\tikzstyle{asym hadamard}=[fill=white,draw,shape=NEbox,inner sep=0.6mm,font=\footnotesize,minimum height=4mm]
\tikzstyle{asym hadamard conj}=[fill=white,draw,shape=NWbox,inner sep=0.6mm,font=\footnotesize,minimum height=4mm]
\tikzstyle{asym hadamard dag}=[fill=white,draw,shape=SEbox,inner sep=0.6mm,font=\footnotesize,minimum height=4mm]

\tikzstyle{hadamard}=[fill=white,draw,inner sep=0.6mm,font=\footnotesize,minimum height=4mm,minimum width=4mm]
\tikzstyle{small hadamard}=[fill=white,draw,inner sep=0.6mm,minimum height=1.5mm,minimum width=1.5mm]
\tikzstyle{small hadamard rotate}=[small hadamard,rotate=45]
\tikzstyle{dhadamard}=[hadamard,doubled]
\tikzstyle{small dhadamard}=[small hadamard,doubled]
\tikzstyle{small dhadamard rotate}=[small hadamard rotate,doubled]
\tikzstyle{antipode}=[white dot,inner sep=0.3mm,font=\footnotesize]

\tikzstyle{scalar}=[diamond,draw,inner sep=0.5pt,font=\small]
\tikzstyle{dscalar}=[diamond,doubled, draw,inner sep=0.5pt,font=\small]

\tikzstyle{small box}=[rectangle,inline text,fill=white,draw,minimum height=5mm,yshift=-0.5mm,minimum width=5mm,font=\small]
\tikzstyle{small gray box}=[small box,fill=gray!30]
\tikzstyle{medium box}=[rectangle,inline text,fill=white,draw,minimum height=5mm,yshift=-0.5mm,minimum width=8mm,font=\small]
\tikzstyle{square box}=[small box] 
\tikzstyle{medium gray box}=[small box,fill=gray!30]
\tikzstyle{semilarge box}=[rectangle,inline text,fill=white,draw,minimum height=5mm,yshift=-0.5mm,minimum width=12.5mm,font=\small]
\tikzstyle{large box}=[rectangle,inline text,fill=white,draw,minimum height=5mm,yshift=-0.5mm,minimum width=15mm,font=\small]
\tikzstyle{large gray box}=[small box,fill=gray!30]

\tikzstyle{Bayes box}=[rectangle,fill=black,draw, minimum height=3mm, minimum width=3mm]

\tikzstyle{gray square point}=[small box,fill=gray!50]

\tikzstyle{dphase box white}=[dhadamard]
\tikzstyle{dphase box gray}=[dhadamard,fill=gray!50!white]
\tikzstyle{phase box white}=[hadamard]
\tikzstyle{phase box gray}=[hadamard,fill=gray!50!white]

\tikzstyle{point}=[regular polygon,regular polygon sides=3,draw,scale=0.75,inner sep=-0.5pt,minimum width=9mm,fill=white,regular polygon rotate=180]
\tikzstyle{copoint}=[regular polygon,regular polygon sides=3,draw,scale=0.75,inner sep=-0.5pt,minimum width=9mm,fill=white]
\tikzstyle{dpoint}=[point,doubled]
\tikzstyle{dcopoint}=[copoint,doubled]

\tikzstyle{wide copoint}=[fill=white,draw,shape=isosceles triangle,shape border rotate=90,isosceles triangle stretches=true,inner sep=0pt,minimum width=1.5cm,minimum height=6.12mm]
\tikzstyle{wide point}=[fill=white,draw,shape=isosceles triangle,shape border rotate=-90,isosceles triangle stretches=true,inner sep=0pt,minimum width=1.5cm,minimum height=6.12mm,yshift=-0.0mm]
\tikzstyle{wide point plus}=[fill=white,draw,shape=isosceles triangle,shape border rotate=-90,isosceles triangle stretches=true,inner sep=0pt,minimum width=1.74cm,minimum height=7mm,yshift=-0.0mm]

\tikzstyle{wide dpoint}=[fill=white,doubled,draw,shape=isosceles triangle,shape border rotate=-90,isosceles triangle stretches=true,inner sep=0pt,minimum width=1.5cm,minimum height=6.12mm,yshift=-0.0mm]

\tikzstyle{tinypoint}=[regular polygon,regular polygon sides=3,draw,scale=0.55,inner sep=-0.15pt,minimum width=6mm,fill=white,regular polygon rotate=180] 

\tikzstyle{white point}=[point]
\tikzstyle{white dpoint}=[dpoint]
\tikzstyle{green point}=[white point] 
\tikzstyle{white copoint}=[copoint]
\tikzstyle{gray point}=[point,fill=gray!40!white]
\tikzstyle{gray dpoint}=[gray point,doubled]
\tikzstyle{red point}=[gray point] 
\tikzstyle{gray copoint}=[copoint,fill=gray!40!white]
\tikzstyle{gray dcopoint}=[gray copoint,doubled]

\tikzstyle{white point guide}=[regular polygon,regular polygon sides=3,font=\scriptsize,draw,scale=0.65,inner sep=-0.5pt,minimum width=9mm,fill=white,regular polygon rotate=180]

\tikzstyle{black point}=[point,fill=black,font=\color{white}]
\tikzstyle{black copoint}=[copoint,fill=black,font=\color{white}]

\tikzstyle{tiny gray point}=[tinypoint,fill=gray!40!white]

\tikzstyle{diredge}=[->]
\tikzstyle{ddiredge}=[<->]
\tikzstyle{rdiredge}=[<-]
\tikzstyle{thickdiredge}=[->, very thick]
\tikzstyle{pointer edge}=[->,very thick,gray]
\tikzstyle{pointer edge part}=[very thick,gray]
\tikzstyle{dashed edge}=[dashed]
\tikzstyle{thick dashed edge}=[very thick,dashed]
\tikzstyle{thick gray dashed edge}=[thick dashed edge,gray!40]
\tikzstyle{thick map edge}=[very thick,|->]

\makeatletter
\newcommand{\boxshape}[3]{%
\pgfdeclareshape{#1}{
\inheritsavedanchors[from=rectangle] 
\inheritanchorborder[from=rectangle]
\inheritanchor[from=rectangle]{center}
\inheritanchor[from=rectangle]{north}
\inheritanchor[from=rectangle]{south}
\inheritanchor[from=rectangle]{west}
\inheritanchor[from=rectangle]{east}
\backgroundpath{
\southwest \pgf@xa=\pgf@x \pgf@ya=\pgf@y
\northeast \pgf@xb=\pgf@x \pgf@yb=\pgf@y

\@tempdima=#2
\@tempdimb=#3

\pgfpathmoveto{\pgfpoint{\pgf@xa - 5pt + \@tempdima}{\pgf@ya}}
\pgfpathlineto{\pgfpoint{\pgf@xa - 5pt - \@tempdima}{\pgf@yb}}
\pgfpathlineto{\pgfpoint{\pgf@xb + 5pt + \@tempdimb}{\pgf@yb}}
\pgfpathlineto{\pgfpoint{\pgf@xb + 5pt - \@tempdimb}{\pgf@ya}}
\pgfpathlineto{\pgfpoint{\pgf@xa - 5pt + \@tempdima}{\pgf@ya}}
\pgfpathclose
}
}}

\boxshape{NEbox}{0pt}{5pt}
\boxshape{SEbox}{0pt}{-5pt}
\boxshape{NWbox}{5pt}{0pt}
\boxshape{SWbox}{-5pt}{0pt}
\boxshape{EBox}{-3pt}{3pt}
\boxshape{WBox}{3pt}{-3pt}
\makeatother

\tikzstyle{cloud}=[shape=cloud,draw,minimum width=1.5cm,minimum height=1.5cm]

\tikzstyle{map}=[draw,shape=NEbox,inner sep=2pt,minimum height=6mm,fill=white]
\tikzstyle{dashedmap}=[draw,dashed,shape=NEbox,inner sep=2pt,minimum height=6mm,fill=white]
\tikzstyle{mapdag}=[draw,shape=SEbox,inner sep=2pt,minimum height=6mm,fill=white]
\tikzstyle{mapadj}=[draw,shape=SEbox,inner sep=2pt,minimum height=6mm,fill=white]
\tikzstyle{maptrans}=[draw,shape=SWbox,inner sep=2pt,minimum height=6mm,fill=white]
\tikzstyle{mapconj}=[draw,shape=NWbox,inner sep=2pt,minimum height=6mm,fill=white]

\tikzstyle{medium map}=[draw,shape=NEbox,inner sep=2pt,minimum height=6mm,fill=white,minimum width=7mm]
\tikzstyle{medium map dag}=[draw,shape=SEbox,inner sep=2pt,minimum height=6mm,fill=white,minimum width=7mm]
\tikzstyle{medium map adj}=[draw,shape=SEbox,inner sep=2pt,minimum height=6mm,fill=white,minimum width=7mm]
\tikzstyle{medium map trans}=[draw,shape=SWbox,inner sep=2pt,minimum height=6mm,fill=white,minimum width=7mm]
\tikzstyle{medium map conj}=[draw,shape=NWbox,inner sep=2pt,minimum height=6mm,fill=white,minimum width=7mm]
\tikzstyle{semilarge map}=[draw,shape=NEbox,inner sep=2pt,minimum height=6mm,fill=white,minimum width=9.5mm]
\tikzstyle{semilarge map trans}=[draw,shape=SWbox,inner sep=2pt,minimum height=6mm,fill=white,minimum width=9.5mm]
\tikzstyle{semilarge map adj}=[draw,shape=SEbox,inner sep=2pt,minimum height=6mm,fill=white,minimum width=9.5mm]
\tikzstyle{semilarge map dag}=[draw,shape=SEbox,inner sep=2pt,minimum height=6mm,fill=white,minimum width=9.5mm]
\tikzstyle{semilarge map conj}=[draw,shape=NWbox,inner sep=2pt,minimum height=6mm,fill=white,minimum width=9.5mm]
\tikzstyle{large map}=[draw,shape=NEbox,inner sep=2pt,minimum height=6mm,fill=white,minimum width=12mm]
\tikzstyle{large map conj}=[draw,shape=NWbox,inner sep=2pt,minimum height=6mm,fill=white,minimum width=12mm]
\tikzstyle{very large map}=[draw,shape=NEbox,inner sep=2pt,minimum height=6mm,fill=white,minimum width=17mm]

\tikzstyle{medium dmap}=[draw,doubled,shape=NEbox,inner sep=2pt,minimum height=6mm,fill=white,minimum width=7mm]
\tikzstyle{medium dmap dag}=[draw,doubled,shape=SEbox,inner sep=2pt,minimum height=6mm,fill=white,minimum width=7mm]
\tikzstyle{medium dmap adj}=[draw,doubled,shape=SEbox,inner sep=2pt,minimum height=6mm,fill=white,minimum width=7mm]
\tikzstyle{medium dmap trans}=[draw,doubled,shape=SWbox,inner sep=2pt,minimum height=6mm,fill=white,minimum width=7mm]
\tikzstyle{medium dmap conj}=[draw,doubled,shape=NWbox,inner sep=2pt,minimum height=6mm,fill=white,minimum width=7mm]
\tikzstyle{semilarge dmap}=[draw,doubled,shape=NEbox,inner sep=2pt,minimum height=6mm,fill=white,minimum width=9.5mm]
\tikzstyle{semilarge dmap trans}=[draw,doubled,shape=SWbox,inner sep=2pt,minimum height=6mm,fill=white,minimum width=9.5mm]
\tikzstyle{semilarge dmap adj}=[draw,doubled,shape=SEbox,inner sep=2pt,minimum height=6mm,fill=white,minimum width=9.5mm]
\tikzstyle{semilarge dmap dag}=[draw,doubled,shape=SEbox,inner sep=2pt,minimum height=6mm,fill=white,minimum width=9.5mm]
\tikzstyle{semilarge dmap conj}=[draw,doubled,shape=NWbox,inner sep=2pt,minimum height=6mm,fill=white,minimum width=9.5mm]
\tikzstyle{large dmap}=[draw,doubled,shape=NEbox,inner sep=2pt,minimum height=6mm,fill=white,minimum width=12mm]
\tikzstyle{large dmap conj}=[draw,doubled,shape=NWbox,inner sep=2pt,minimum height=6mm,fill=white,minimum width=12mm]
\tikzstyle{large dmap trans}=[draw,doubled,shape=SWbox,inner sep=2pt,minimum height=6mm,fill=white,minimum width=12mm]
\tikzstyle{large dmap adj}=[draw,doubled,shape=SEbox,inner sep=2pt,minimum height=6mm,fill=white,minimum width=12mm]
\tikzstyle{large dmap dag}=[draw,doubled,shape=SEbox,inner sep=2pt,minimum height=6mm,fill=white,minimum width=12mm]
\tikzstyle{very large dmap}=[draw,doubled,shape=NEbox,inner sep=2pt,minimum height=6mm,fill=white,minimum width=19.5mm]

\tikzstyle{muxbox}=[draw,shape=rectangle,minimum height=3mm,minimum width=3mm,fill=white]
\tikzstyle{dmuxbox}=[muxbox,doubled]

\tikzstyle{box}=[draw,shape=rectangle,inner sep=2pt,minimum height=6mm,minimum width=6mm,fill=white]
\tikzstyle{dbox}=[draw,doubled,shape=rectangle,inner sep=2pt,minimum height=6mm,minimum width=6mm,fill=white]
\tikzstyle{dmap}=[draw,doubled,shape=NEbox,inner sep=2pt,minimum height=6mm,fill=white]
\tikzstyle{dmapdag}=[draw,doubled,shape=SEbox,inner sep=2pt,minimum height=6mm,fill=white]
\tikzstyle{dmapadj}=[draw,doubled,shape=SEbox,inner sep=2pt,minimum height=6mm,fill=white]
\tikzstyle{dmaptrans}=[draw,doubled,shape=SWbox,inner sep=2pt,minimum height=6mm,fill=white]
\tikzstyle{dmapconj}=[draw,doubled,shape=NWbox,inner sep=2pt,minimum height=6mm,fill=white]

\tikzstyle{ddmap}=[draw,doubled,dashed,shape=NEbox,inner sep=2pt,minimum height=6mm,fill=white]
\tikzstyle{ddmapdag}=[draw,doubled,dashed,shape=SEbox,inner sep=2pt,minimum height=6mm,fill=white]
\tikzstyle{ddmapadj}=[draw,doubled,dashed,shape=SEbox,inner sep=2pt,minimum height=6mm,fill=white]
\tikzstyle{ddmaptrans}=[draw,doubled,dashed,shape=SWbox,inner sep=2pt,minimum height=6mm,fill=white]
\tikzstyle{ddmapconj}=[draw,doubled,dashed,shape=NWbox,inner sep=2pt,minimum height=6mm,fill=white]

\boxshape{sNEbox}{0pt}{3pt}
\boxshape{sSEbox}{0pt}{-3pt}
\boxshape{sNWbox}{3pt}{0pt}
\boxshape{sSWbox}{-3pt}{0pt}
\tikzstyle{smap}=[draw,shape=sNEbox,fill=white]
\tikzstyle{smapdag}=[draw,shape=sSEbox,fill=white]
\tikzstyle{smapadj}=[draw,shape=sSEbox,fill=white]
\tikzstyle{smaptrans}=[draw,shape=sSWbox,fill=white]
\tikzstyle{smapconj}=[draw,shape=sNWbox,fill=white]

\tikzstyle{dsmap}=[draw,dashed,shape=sNEbox,fill=white]
\tikzstyle{dsmapdag}=[draw,dashed,shape=sSEbox,fill=white]
\tikzstyle{dsmaptrans}=[draw,dashed,shape=sSWbox,fill=white]
\tikzstyle{dsmapconj}=[draw,dashed,shape=sNWbox,fill=white]

\boxshape{mNEbox}{0pt}{10pt}
\boxshape{mSEbox}{0pt}{-10pt}
\boxshape{mNWbox}{10pt}{0pt}
\boxshape{mSWbox}{-10pt}{0pt}
\tikzstyle{mmap}=[draw,shape=mNEbox]
\tikzstyle{mmapdag}=[draw,shape=mSEbox]
\tikzstyle{mmaptrans}=[draw,shape=mSWbox]
\tikzstyle{mmapconj}=[draw,shape=mNWbox]

\tikzstyle{mmapgray}=[draw,fill=gray!40!white,shape=mNEbox]
\tikzstyle{smapgray}=[draw,fill=gray!40!white,shape=sNEbox]

\makeatletter

\pgfdeclareshape{cornerpoint}{
\inheritsavedanchors[from=rectangle] 
\inheritanchorborder[from=rectangle]
\inheritanchor[from=rectangle]{center}
\inheritanchor[from=rectangle]{north}
\inheritanchor[from=rectangle]{south}
\inheritanchor[from=rectangle]{west}
\inheritanchor[from=rectangle]{east}
\backgroundpath{
\southwest \pgf@xa=\pgf@x \pgf@ya=\pgf@y
\northeast \pgf@xb=\pgf@x \pgf@yb=\pgf@y

\pgfmathsetmacro{\pgf@shorten@left}{\pgfkeysvalueof{/tikz/shorten left}}
\pgfmathsetmacro{\pgf@shorten@right}{\pgfkeysvalueof{/tikz/shorten right}}

\pgfpathmoveto{\pgfpoint{0.5 * (\pgf@xa + \pgf@xb)}{\pgf@ya - 5pt}}
\pgfpathlineto{\pgfpoint{\pgf@xa - 8pt + \pgf@shorten@left}{\pgf@yb - 1.5 * \pgf@shorten@left}}
\pgfpathlineto{\pgfpoint{\pgf@xa - 8pt + \pgf@shorten@left}{\pgf@yb}}
\pgfpathlineto{\pgfpoint{\pgf@xb + 8pt - \pgf@shorten@right}{\pgf@yb}}
\pgfpathlineto{\pgfpoint{\pgf@xb + 8pt - \pgf@shorten@right}{\pgf@yb - 1.5 * \pgf@shorten@right}}
\pgfpathclose
}
}

\pgfdeclareshape{cornercopoint}{
\inheritsavedanchors[from=rectangle] 
\inheritanchorborder[from=rectangle]
\inheritanchor[from=rectangle]{center}
\inheritanchor[from=rectangle]{north}
\inheritanchor[from=rectangle]{south}
\inheritanchor[from=rectangle]{west}
\inheritanchor[from=rectangle]{east}
\backgroundpath{
\southwest \pgf@xa=\pgf@x \pgf@ya=\pgf@y
\northeast \pgf@xb=\pgf@x \pgf@yb=\pgf@y

\pgfmathsetmacro{\pgf@shorten@left}{\pgfkeysvalueof{/tikz/shorten left}}
\pgfmathsetmacro{\pgf@shorten@right}{\pgfkeysvalueof{/tikz/shorten right}}

\pgfpathmoveto{\pgfpoint{0.5 * (\pgf@xa + \pgf@xb)}{\pgf@yb + 5pt}}
\pgfpathlineto{\pgfpoint{\pgf@xa - 8pt + \pgf@shorten@left}{\pgf@ya + 1.5 * \pgf@shorten@left}}
\pgfpathlineto{\pgfpoint{\pgf@xa - 8pt + \pgf@shorten@left}{\pgf@ya}}
\pgfpathlineto{\pgfpoint{\pgf@xb + 8pt - \pgf@shorten@right}{\pgf@ya}}
\pgfpathlineto{\pgfpoint{\pgf@xb + 8pt - \pgf@shorten@right}{\pgf@ya + 1.5 * \pgf@shorten@right}}
\pgfpathclose
}
}

\makeatother

\pgfkeyssetvalue{/tikz/shorten left}{0pt}
\pgfkeyssetvalue{/tikz/shorten right}{0pt}

\tikzstyle{kpoint common}=[draw,fill=white,inner sep=1pt,minimum height=4mm]
\tikzstyle{kpoint sc}=[shape=cornerpoint,kpoint common]
\tikzstyle{kpoint adjoint sc}=[shape=cornercopoint,kpoint common]
\tikzstyle{kpoint}=[shape=cornerpoint,shorten left=5pt,kpoint common]
\tikzstyle{kpoint adjoint}=[shape=cornercopoint,shorten left=5pt,kpoint common]
\tikzstyle{kpoint conjugate}=[shape=cornerpoint,shorten right=5pt,kpoint common]
\tikzstyle{kpoint transpose}=[shape=cornercopoint,shorten right=5pt,kpoint common]
\tikzstyle{kpoint symm}=[shape=cornerpoint,shorten left=5pt,shorten right=5pt,kpoint common]

\tikzstyle{black kpoint}=[shape=cornerpoint,shorten left=5pt,kpoint common,fill=black,font=\color{white}]
\tikzstyle{black kpoint adjoint}=[shape=cornercopoint,shorten left=5pt,kpoint common,fill=black,font=\color{white}]
\tikzstyle{black kpointadj}=[shape=cornercopoint,shorten left=5pt,kpoint common,fill=black,font=\color{white}]

\tikzstyle{black dkpoint}=[shape=cornerpoint,shorten left=5pt,kpoint common,fill=black, doubled,font=\color{white}]
\tikzstyle{black dkpoint adjoint}=[shape=cornercopoint,shorten left=5pt,kpoint common,fill=black, doubled,font=\color{white}]
\tikzstyle{black dkpointadj}=[shape=cornercopoint,shorten left=5pt,kpoint common,fill=black, doubled,font=\color{white}] 

\tikzstyle{kpointdag}=[kpoint adjoint]
\tikzstyle{kpointadj}=[kpoint adjoint]
\tikzstyle{kpointconj}=[kpoint conjugate]
\tikzstyle{kpointtrans}=[kpoint transpose]

\tikzstyle{big kpoint}=[kpoint, minimum width=1.2 cm, minimum height=8mm, inner sep=4pt, text depth=3mm]

\tikzstyle{wide kpoint}=[kpoint, minimum width=1 cm, inner sep=2pt]
\tikzstyle{wide kpointdag}=[kpointdag, minimum width=1 cm, inner sep=2pt]
\tikzstyle{wide kpointconj}=[kpointconj, minimum width=1 cm, inner sep=2pt]
\tikzstyle{wide kpointtrans}=[kpointtrans, minimum width=1 cm, inner sep=2pt]

\tikzstyle{gray kpoint}=[kpoint,fill=gray!50!white]
\tikzstyle{gray kpointdag}=[kpointdag,fill=gray!50!white]
\tikzstyle{gray kpointadj}=[kpointadj,fill=gray!50!white]
\tikzstyle{gray kpointconj}=[kpointconj,fill=gray!50!white]
\tikzstyle{gray kpointtrans}=[kpointtrans,fill=gray!50!white]

\tikzstyle{gray dkpoint}=[kpoint,fill=gray!50!white,doubled]
\tikzstyle{gray dkpointdag}=[kpointdag,fill=gray!50!white,doubled]
\tikzstyle{gray dkpointadj}=[kpointadj,fill=gray!50!white,doubled]
\tikzstyle{gray dkpointconj}=[kpointconj,fill=gray!50!white,doubled]
\tikzstyle{gray dkpointtrans}=[kpointtrans,fill=gray!50!white,doubled]

\tikzstyle{white label}=[draw,fill=white,rectangle,inner sep=0.7 mm]
\tikzstyle{gray label}=[draw,fill=gray!50!white,rectangle,inner sep=0.7 mm]
\tikzstyle{black label}=[draw,fill=black,rectangle,inner sep=0.7 mm]

\tikzstyle{dkpoint}=[kpoint,doubled]
\tikzstyle{wide dkpoint}=[wide kpoint,doubled]
\tikzstyle{dkpointdag}=[kpoint adjoint,doubled]
\tikzstyle{wide dkpointdag}=[wide kpointdag,doubled]
\tikzstyle{dkcopoint}=[kpoint adjoint,doubled]
\tikzstyle{dkpointadj}=[kpoint adjoint,doubled]
\tikzstyle{dkpointconj}=[kpoint conjugate,doubled]
\tikzstyle{dkpointtrans}=[kpoint transpose,doubled]

\tikzstyle{kscalar}=[kpoint common, shape=EBox, inner xsep=-1pt, inner ysep=3pt,font=\small]
\tikzstyle{kscalarconj}=[kpoint common, shape=WBox, inner xsep=-1pt, inner ysep=3pt,font=\small]

\tikzstyle{spekpoint}=[kpoint sc,minimum height=5mm,inner sep=3pt]
\tikzstyle{spekcopoint}=[kpoint adjoint sc,minimum height=5mm,inner sep=3pt]

\tikzstyle{dspekpoint}=[spekpoint,doubled]
\tikzstyle{dspekcopoint}=[spekcopoint,doubled]

 \tikzstyle{upground}=[circuit ee IEC,thick,ground,rotate=90,scale=1.5,inner sep=-2mm]
 \tikzstyle{downground}=[circuit ee IEC,thick,ground,rotate=-90,scale=1.5,inner sep=-2mm]

\tikzstyle{maxmix}=[regular polygon,regular polygon sides=3,draw=black,xscale=0.4,yscale=0.3,inner sep=-0.5pt,minimum width=10mm,fill=gray,regular polygon rotate=180]

 \tikzstyle{bigground}=[regular polygon,regular polygon sides=3,draw=gray,scale=0.50,inner sep=-0.5pt,minimum width=10mm,fill=gray]

\tikzstyle{arrs}=[-latex,font=\small,auto]
\tikzstyle{arrow plain}=[arrs]
\tikzstyle{arrow dashed}=[dashed,arrs]
\tikzstyle{arrow bold}=[very thick,arrs]
\tikzstyle{arrow hide}=[draw=white!0,-]
\tikzstyle{arrow reverse}=[latex-]
\tikzstyle{cdnode}=[]

%% file: figures/seq-comp.tikz
\begin{tikzpicture}
	\begin{pgfonlayer}{nodelayer}
		\node [style={small box}] (0) at (0, 0.75) {$g$};
		\node [style={small box}] (1) at (0, -0.75) {$f$};
		\node [style=none] (2) at (0, 1.75) {};
		\node [style=none] (3) at (0, -1.75) {};
	\end{pgfonlayer}
	\begin{pgfonlayer}{edgelayer}
		\draw (3.center) to (2.center);
	\end{pgfonlayer}
\end{tikzpicture}

%% file: figures/wire.tikz
\begin{tikzpicture}
	\begin{pgfonlayer}{nodelayer}
		\node [style=none] (0) at (0, 0.75) {};
		\node [style=none] (1) at (0, -0.75) {};
		\node [style=right label] (2) at (0.25, 0) {$A$};
	\end{pgfonlayer}
	\begin{pgfonlayer}{edgelayer}
		\draw (1.center) to (0.center);
	\end{pgfonlayer}
\end{tikzpicture}

%% file: figures/cup.tikz
\begin{tikzpicture}
	\begin{pgfonlayer}{nodelayer}
		\node [style=none] (0) at (-0.75, 0.5) {};
		\node [style=none] (1) at (0.75, 0.5) {};  
	\end{pgfonlayer}
	\begin{pgfonlayer}{edgelayer}
		\draw [in=-90, out=-90, looseness=1.75] (0.center) to (1.center);    
	\end{pgfonlayer}
\end{tikzpicture}

%% file: figures/cap.tikz
\begin{tikzpicture}
	\begin{pgfonlayer}{nodelayer}
		\node [style=none] (0) at (-0.75, -0.5) {};
		\node [style=none] (1) at (0.75, -0.5) {};
	\end{pgfonlayer}
	\begin{pgfonlayer}{edgelayer}
		\draw [in=90, out=90, looseness=1.75] (0.center) to (1.center);  
	\end{pgfonlayer}
\end{tikzpicture}

%% file: figures/causal-graph-AB.tikz
\begin{tikzpicture}
	\path [use as bounding box] (-1,-2) rectangle (1,2);
	\begin{pgfonlayer}{nodelayer}
		\node [style=label, inner sep={0.5 mm}] (0) at (0, -1) {\textsf{A}};
		\node [style=label, inner sep={0.5 mm}] (1) at (0, 1) {\textsf{B}};
	\end{pgfonlayer}
	\begin{pgfonlayer}{edgelayer}
		\draw (0) to (1);
	\end{pgfonlayer}
\end{tikzpicture}

%% file: figures/causal-graph-BA.tikz
\begin{tikzpicture}
	\path [use as bounding box] (-1,-2) rectangle (1,2);
	\begin{pgfonlayer}{nodelayer}
		\node [style=label, inner sep={0.5 mm}] (0) at (0, 1) {\textsf{A}};
		\node [style=label, inner sep={0.5 mm}] (1) at (0, -1) {\textsf{B}};
	\end{pgfonlayer}
	\begin{pgfonlayer}{edgelayer}
		\draw (0) to (1);
	\end{pgfonlayer}
\end{tikzpicture}

%% file: figures/causal-graph-ns.tikz
\begin{tikzpicture}
	\path [use as bounding box] (-1.75,-1.25) rectangle (1.75,1.25);
	\begin{pgfonlayer}{nodelayer}
		\node [style=label, inner sep={0.5 mm}] (0) at (-0.75, -0) {\textsf{A}};
		\node [style=label, inner sep={0.5 mm}] (1) at (0.75, -0) {\textsf{B}};
	\end{pgfonlayer}
\end{tikzpicture}

%% file: figures/dim.tikz
\begin{tikzpicture}
	\begin{pgfonlayer}{nodelayer}
		\node [style=upground] (0) at (0, 0.75) {};
		\node [style=downground] (1) at (0, -0.75) {};
		\node [style=right label] (2) at (0.25, 0) {$A$};
	\end{pgfonlayer}
	\begin{pgfonlayer}{edgelayer}
		\draw (1) to (0);
	\end{pgfonlayer}
\end{tikzpicture}

%% file: figures/dual-f1.tikz
\begin{tikzpicture}
	\begin{pgfonlayer}{nodelayer}
		\node [style=copoint] (0) at (0, 1.5) {$\pi$};
		\node [style=point] (1) at (0, -1.5) {$\rho$};
		\node [style=small box] (2) at (0, 0) {$f$};
	\end{pgfonlayer}
	\begin{pgfonlayer}{edgelayer}
		\draw (0) to (1);
	\end{pgfonlayer}
\end{tikzpicture}

%% file: figures/common-cause-lin1.tikz
\begin{tikzpicture}
	\begin{pgfonlayer}{nodelayer}
		\node [style=label, inner sep={0.5 mm}] (0) at (0, -2) {\textsf{A}};
		\node [style=label, inner sep={0.5 mm}] (1) at (0, 0) {\textsf{B}};
		\node [style=label, inner sep={0.5 mm}] (2) at (0, 2) {\textsf{C}};
	\end{pgfonlayer}
	\begin{pgfonlayer}{edgelayer}
		\draw (0) to (1);
		\draw (1) to (2);
	\end{pgfonlayer}
\end{tikzpicture}

%% file: figures/common-cause-lin2.tikz
\begin{tikzpicture}
	\begin{pgfonlayer}{nodelayer}
		\node [style=label, inner sep={0.5 mm}] (0) at (0, 2) {\textsf{B}};
		\node [style=label, inner sep={0.5 mm}] (1) at (0, -0) {\textsf{C}};
		\node [style=label, inner sep={0.5 mm}] (2) at (0, -2) {\textsf{A}};
	\end{pgfonlayer}
	\begin{pgfonlayer}{edgelayer}
		\draw (2) to (1);
		\draw (1) to (0);
	\end{pgfonlayer}
\end{tikzpicture}